\documentclass[11pt]{article}

\usepackage[T1]{fontenc}
\usepackage{lmodern}

\usepackage[utf8]{inputenc}
\usepackage[english]{babel}
\usepackage{amsmath}
\usepackage{nicefrac}
\usepackage{amsthm}
\usepackage{amsfonts}
\usepackage{amssymb}
\usepackage{verbatim}
\usepackage{color}
\usepackage[arrow, matrix, curve]{xy}
\usepackage{bbm}
\usepackage{eqparbox}
\usepackage[numbers]{natbib}
\usepackage{stmaryrd}
\usepackage{amssymb}
\usepackage{mathrsfs}
\usepackage{pictexwd,dcpic}
\usepackage{mathabx}

\newtheoremstyle{WreschTheoremstyle} 
                        {1.5em}    
                        {2.5em}    
                        {}         
                        {}         
                        {\bfseries}
                        {}        
                        {\newline} 
                        {\raisebox{0.6em}{\thmname{#1}\thmnumber{#2}\thmnote{ (#3)}}}

\newcommand{\R}{\mathbb{R}}

\newcommand{\N}{\mathbb{N}}
\newcommand{\Z}{\mathbb{Z}}

\newcommand{\F}{\mathcal{F}}

\newcommand{\E}{\mathbb{E}}
\newcommand{\e}{\varepsilon}

\renewcommand{\1}{\mathbbm{1}}
\renewcommand{\Pr}{\mathbbm{P}}

\newtheorem{Theorem}{Theorem}[section]
\newtheorem{Proposition}[Theorem]{Proposition}
\newtheorem{Corollary}[Theorem]{Corollary}
\newtheorem{Lemma}[Theorem]{Lemma}
\newtheorem{Remark}[Theorem]{Remark}
\newtheorem{Definition}[Theorem]{Definition}
\newtheorem{Example}[Theorem]{Example}

\numberwithin{equation}{section}

\makeatletter
\newcommand{\customlabel}[1]{%
     \stepcounter{ref}%
   \protected@write
\@auxout{}{\string\newlabel{#1}{{\thesatz.\arabic{ref}}{\thepage}{\thesatz.\arabic{ref}}{#1}{}}}%
   \hypertarget{#1}{\thesatz.\arabic{ref}}%
}
\makeatother

\newenvironment{sciabstract}{\begin{quote}}{\end{quote}}

\newcounter{lastnote}

\title{The Enskog process for hard and soft potentials}
\newcommand{\pdftitle} {The Enskog process for hard and soft potentials}
\newcommand{\pdfauthor}{Martin Friesen}

\author{
Martin Friesen\footnote{Deparment of Mathematics, Wuppertal University, Gaußstraße 20, 42119 Wuppertal, Germany, friesen@math.uni-wuppertal.de}\\
Barbara R\"udiger\footnote{Deparment of Mathematics, Wuppertal University, Gaußstraße 20, 42119 Wuppertal, Germany, ruediger@uni-wuppertal.de}\\
Padmanabhan Sundar\footnote{Department of Mathematics, Louisiana State University, Baton Rouge, Louisiana 70803, USA, psundar@lsu.edu}
}

\usepackage[plainpages=false,pdfpagelabels=true,bookmarks=true,pdfauthor={\pdfauthor},
pdftitle={\pdftitle}]{hyperref}

\makeatletter
\def\HyPsd@CatcodeWarning#1{}
\makeatother

\begin{document}

\maketitle
\begin{sciabstract}\textbf{Abstract:}
The density of a moderately dense gas evolving in a vacuum is given by the solution of an Enskog equation.
Recently we have constructed in \cite{ARS17} the stochastic process that corresponds to the Enskog equation under suitable conditions.
The Enskog process is identified as the solution of a McKean-Vlasov equation driven by a Poisson random measure.
In this work, we continue the study for a wider class of collision kernels that includes hard and soft potentials.
Based on a suitable particle approximation of binary collisions, the existence of an Enskog process is established.
\end{sciabstract}


\section{Introduction}

\subsection{The Boltzmann-Enskog equation}
In kinetic theory of gases each particle is completely described by its position $r \in \R^d$ and its velocity $v \in \R^d$,
where $d \geq 3$. It moves with constant speed $v$ until it performs a collision with another particle $(q,u)$.
Denote by $v^{\star},u^{\star}$ the resulting velocities after collision.
We suppose that collisions are elastic, as a consequence conservation of momentum and kinetic energy hold, i.e.
\begin{align}\label{CONSERVATION1}
  u + v &= u^{\star} + v^{\star} \\ |u|^2 + |v|^2 &= |u^{\star}|^2 + |v^{\star}|^2.
\end{align}
A commonly used parameterization of the deflected velocities $v^{\star},u^{\star}$ is given by the vector $n = \frac{v^{\star} - v}{|v^{\star} - v|}$ via
\begin{align}\label{COLLISION}
 \begin{cases} v^{\star} &= v + (u-v,n)n \\ u^{\star} &= u - (u-v,n)n \end{cases}, \ \ n \in S^{d-1},
\end{align}
where $(\cdot,\cdot)$ denotes the euclidean product in $\R^d$.
In the physical sense, $n$ is parallel to the deflected velocity 
$v^{\star} - v$ as well as
to the segment joining the centers of the spheres at
the instant of collision, see, e.g. Bressan \cite{B05}.
Note that, for fixed $n \in S^{d-1}$, the change of variables $(v,u) \longmapsto (v^{\star}, u^{\star})$
is an involutive transformation with Jacobian equal to $1$.

Let $f_0(r,v) \geq 0$ be the particle density function of the gas at initial time $t = 0$.
The time evolution $f_t = f_t(r,v)$ is then obtained from the (Boltzmann-)Enskog equation
\begin{align}\label{EQ:03}
 \frac{\partial f_t}{\partial t} + v \cdot (\nabla_r f_t) = \mathcal{Q}(f_t,f_t), \ \ f_t|_{t=0} = f_0, \ \ t > 0.
\end{align}
Here $\mathcal{Q}$ is a non-local, nonlinear collision integral operator given by
\begin{align}\label{CINT}
 \mathcal{Q}(f_t,f_t)(r,v) = \int \limits_{\R^{2d}}\int \limits_{S^{d-1}}\left( f_t(r,v^{\star})f_t(q,u^{\star}) - f_t(r,v)f_t(q,u)\right)\beta(r-q)B(|v-u|,n)dn du dq,
\end{align}
where $dn$ denotes the Lebesgue surface measure on the sphere $S^{d-1}$ and $B(|v-u|, n) \geq 0$ the collision kernel.
The particular form of $B(|v-u|,n)$ depends on the particular microscopic model one has in mind.
The function $\beta \geq 0$ describes the rate at which a particle at position $r$ performs a collision with another particle at position $q$.
Concerning applications and additional physical background on this topic the reader may consult the classical books of Cercignani \cite{C88} and Cercignani, Illner, Pulvirenti \cite{CIP94}. 
More recent results (and related kinetic equations) are discussed, e.g., in the review articles by Villani \cite{V02} and
Alexandre \cite{A09}, see also the references therein.

\subsection{Typical collision kernels}
Let us briefly comment on particular examples of collision kernels $B(|v-u|,n)$ in dimension $d = 3$.
Boltzmann's original model was first formulated 
for (true) hard spheres where 
\[
 B(|u-v|,n)dn = |(u-v,n)|dn.
\]
A transformation in polar coordinates to a system where the center is in $\frac{u+v}{2}$ and $e_3 = (0,0,1)$ is parallel to 
$u-v$, i.e. $e_3 |u-v| = u-v$, leads to 
\[
 B(|u-v|,n)dn = |(u-v,n)|dn = |u-v| \sin\left( \frac{\theta}{2}\right) \cos\left( \frac{\theta}{2}\right)d\theta d\phi,
\]
where $\theta \in (0,\pi]$ is the angle between $u-v$ and $u^{\star} - v^{\star}$ and $\phi \in (0,2\pi]$ is the longitude angle,
see Tanaka \cite{T78} or Horowitz and Karandikar \cite{HK90}.
More generally, many results rely on Grad's angular cut-off assumption where it is supposed that
\[
 \int \limits_{S^{d-1}}B(|v-u|, n) dn < \infty,
\]
Note that this includes Boltzmanns original model of hard spheres.
However, there exist also several models, where Grad's angular cut-off assumption is not satisfied.
The most prominent ones are long-range interactions described below.
\begin{Example}
 Consider a collision kernel given by
  \begin{align}\label{EQ}
    B(|v-u|, n)dn = |v-u|^{\gamma} b(\theta)d\theta d\xi,
  \end{align}
  where $b$ is at least locally bounded on $(0,\pi]$ and
  \[
    \gamma > -3, \ \  b(\theta) \sim \theta^{-1 - \nu}, \ \ \theta \to 0^+, \  \ \nu \in (0,2).
  \]
  The parameters $\gamma$ and $\nu$ are related by
 \begin{align}\label{INTRO:00}
  \gamma = \frac{s - 5}{s-1}, \qquad \nu = \frac{2}{s-1}, \ \ \ s > 2.
 \end{align}
 One distinguishes between the following cases:
 \begin{enumerate}
  \item[(i)] Very soft potentials $s \in (2,3]$, $\gamma \in (-3,-1]$ and $\nu \in [1,2)$.
  \item[(ii)] Soft potentials $s \in (3,5)$, $\gamma \in (-1, 0)$ and $\nu \in ( \frac{1}{2}, 1)$.
  \item[(iii)] Maxwellian molecules $s = 5$, $\gamma = 0$ and $\nu = \frac{1}{2}$.
  \item[(iv)] Hard potentials $s > 5$, $\gamma \in (0,1)$ and $\nu \in (0,\frac{1}{2})$.
 \end{enumerate}
  For additional details and comments we refer to \cite{V02} or \cite{A09}.
  Note that one has 
 \[
  \int \limits_{0}^{\pi}b(\theta)d\theta = \infty \ \ \text{ but } \ \ \int \limits_{0}^{\pi}\theta^2 b(\theta)d\theta < \infty,
 \]
  i.e. this example does not satisfy Grad's angular cut-off assumption.
  Hence $\mathcal{Q}$ is a nonlinear and singular integral operator with either unbounded or singular coefficients.
  A rigorous analysis of the corresponding Cauchy problem \eqref{EQ:03} is therefore a challenging mathematical task.
\end{Example}

\subsection{The role of $\beta$}
The Cauchy problem \eqref{EQ:03} strongly depends on the particular choice of $\beta$. 
Below we describe some physically different regimes which are typically studied by different techniques.

 \textbf{ Case } $\beta(x-y) \equiv 1$.
 If $f_t$ is a solution to \eqref{EQ:03},
 it is simple to verify that the function $g_t(v) := \int_{\R^d}f_t(r,v)dr$ solves the spatially homogeneous Boltzmann equation
\begin{align}\label{INTRO:01}
 \frac{\partial g_t(v)}{\partial t} = \int \limits_{\R^{d}}\int \limits_{S^{d-1}}\left(g_t(v^{\star})g_t(u^{\star}) - g_t(v)g_t(u)\right)B(|v-u|,n)dn du.
\end{align}
In particular, any two particles, independent of their positions, may perform a collision. 
Such an equation has been intensively studied and a satisfactory theory developed (see e.g. Tanaka \cite{T78, T87}, Arkeryd \cite{A83}, Wennberg \cite{W94}, Villani \cite{V98}, Toscani, Villani \cite{TV99}, Alexandre, Villani \cite{AV02},  Mouhot, Villani \cite{MV04}, Desvillettes, Mouhot \cite{DM09} and Morimoto, Wang, Yang \cite{MWY16}).

 \textbf{ Case } $\beta(x-y) \equiv \delta_0(|x-y|)$ (dirac distribution at zero).
Here we formally recover the classical Boltzmann equation where colliding particles have to be at the same position.
This equation provides a successful description of a dilute gas and can be derived in the Boltzmann-Grad limit from Hamiltonian dynamics
(see Illner, Pulvirenti \cite{IP89}). In this case the collision integral \eqref{CINT} is local (but singular) in the spatial variables.
Classical results on the Boltzmann equation can be found in \cite{V02} and \cite{A09}.
Recently there has been some interesting progress on global solutions close to equilibrium (see e.g. the works of 
Alexandre, Morimoto,  Ukai, Xu, Yang \cite{AMUX11,AMUX11b, AMUX11c, AMUX12}).
Note that in contrast to our work, the solutions studied in the above references are, in general, not probability distributions on $\R^{2d}$.

 \textbf{ Case } $\beta(x-y) \equiv \delta_{\rho}(|x-y|)$, $\rho > 0$ fixed.
In this case particles are described by balls of a fixed radius $\rho > 0$ performing elastic collisions.
Here the collision integral \eqref{CINT} is less singular than in the classical Boltzmann equation.
The corresponding Cauchy problem was studied e.g. by Toscani, Bellomo \cite{TB87}, Arkeryd \cite{A90} and Arkeryd, Cercignani \cite{AC90}.
Based on an interacting particle system of binary collisions, the Boltzmann-Grad limit was established for true hard spheres by Rezakhanlou \cite{R03}.
Most of the results obtained in this direction are mainly applicable under Grad's angular cut-off assumption.

 \textbf{ Case } $0 \leq \beta \in C_c^1(\R^d)$ is symmetric.
This case can be seen as a mollified version of either $\delta_0(|x-y|)$ or $\delta_{\rho}(|x-y|)$.
As a consequence the collision integral \eqref{CINT} is not singular in the spatial variables which allows one to use stochastic methods in the treatment of this model.
The analysis of the corresponding Cauchy problem was initiated by Povzner \cite{P62} under Grad's angular cut-off assumption.
Corresponding propagation of chaos was studied by Cercignani \cite{C83} under Grad's angular cut-off assumption 
performing the Boltzmann-Grad limit for the corresponding BBGKY-hierarchy.
First results applicable without cut-off (including the case of Maxwellian molecules)
 have been recently obtained by Albeverio, R\"udiger, Sundar in \cite{ARS17}, where also the corresponding stochastic process
(the so-called Enskog process) was studied.
In this work we extend the obtained existence result including now also the case of general long-range interactions,
but also the case of true hard spheres as studied in the non-mollified settings.

\subsection{The Enskog process}
Note that any solution $f_t$ to \eqref{EQ:03} is expected to satisfy the conservation laws
\begin{align}\label{CONSERVATION}
\int \limits_{\R^{2d}}\begin{pmatrix}1 \\ v \\ |v|^2 \end{pmatrix} f_t(r,v)dr dv
 = \int \limits_{\R^{2d}}\begin{pmatrix}1 \\ v \\ |v|^2 \end{pmatrix} f_0(r,v)dr dv
\end{align}
and hence preserves, in particular, probability. 
One natural question already posed by Marc Kac \cite{K56} is related with the construction of a stochastic process 
(the Boltzmann process) having time-marginals $f_t$. 
Let us stress that the Boltzmann equation contains only information of the time-marginals at a given time $t$,
while the corresponding Boltzmann process provides a pathwise description and therefore also
contains additional information such as finite-dimensional distributions.

In the particular case of the space-homogeneous Boltzmann equation such a construction was intensively studied in the past.
In his pioneering works Tanaka \cite{T78, T87} has studied for Maxwellian molecules a stochastic process 
for which its time-marginals solve the space-homogeneous Boltzmann equation \eqref{INTRO:01}.
A particle approximation and related propagation of chaos was then established under the same conditions by Horowitz and Karandikar \cite{HK90}.
The construction of the space-homogeneous Boltzmann process and existence of densities has been recently studied by Fournier \cite{F15} for the case
of long-range interactions.
Corresponding particle approximations (including a rate of convergence) was studied by Fournier, Mischler \cite{FM16} for hard potentials
and by Xu \cite{X16} for soft potentials.
Similar results for Maxwellian molecules, but with another particle system, were also obtained in \cite{CF18}.
The precise formulation of hard and soft potentials are introduced in the next section.

The martingale problem associated to the classical Boltzmann equation and related particle approximation 
was studied under Grad's angular cut-off assumption by M\`{e}l\`{e}ard \cite{M98} with sub-gaussian initial distribution,
where convergence of solutions when $\beta \longrightarrow \delta_0$ was also studied.
The existence of a stochastic process associated to the Enskog equation (in the space-inhomogeneous setting) without Grad's angular cutoff assumption was recently obtained for the Enskog equation (with $\beta \in C_c^1(\R^{2d}))$ in \cite{ARS17},
where the corresponding Enskog process was obtained from a McKean-Vlasov stochastic equation with jumps.
In this work we take $\beta \in C_c^1(\R^{2d})$ and impose certain conditions on the collision kernel $B$ which includes
the case of long-range interactions. The main contribution of this work are
\begin{enumerate}
 \item[(i)] Construction and uniqueness (in law) of a physically motivated $n$-particle process with binary collisions.
 \item[(ii)] Identification of an Enskog process obtained from the particle approximation when $n \to \infty$.
\end{enumerate}
As a consequence of our results, uniqueness for the Enskog equation implies propagation of chaos in the sense of Sznitmann \cite{S91}.
The corresponding uniqueness problem will be studied in a separate work.
In the next section we introduce the main objects of this work, while our main results are formulated in Section 3.

\section{Preliminaries}

\subsection{Change of variables for binary collisions}
It was already pointed out by Tanaka that in $d = 3$ 
$(u,v) \longmapsto (u-v,n)n$ cannot be smooth (see previous section).
To overcome this problem he introduced 
in \cite{T78} another transformation of parameters which is bijective,
has jacobian $1$ and hence can be used in the right hand side of 
\eqref{CINT}. Such ideas have been extended to arbitrary dimension $d \geq 3$ and are briefly summarized in this section, see \cite{FM09, LM12}.

Given $u,v \in \R^d$, recall that $v^{\star},u^{\star}$ are computed from \eqref{COLLISION}
where $n = \frac{v^{\star} - v}{|v^{\star} - v|}$. Suppose first that $u \neq v$.
In this work we use another representation of \eqref{COLLISION} where $n$ is decomposed as $n = n_0 + n_1$ with $n_0$ being parallel to $u-v$ and $n_1$ orthogonal to $u-v$.
For this purpose let 
\[
 S^{d-2}(u-v) = \{ \omega \in \R^d \ | \ |u-v| = |\omega|, \ \ (u-v, \omega) = 0\}, \qquad S^{d-2} = \{ \xi \in \R^{d-1} \ | \ |\xi| = 1 \}.
\]
The precise construction of such a parameterization is given in the next lemma.
\begin{Lemma}
 Let $u,v \in \R^d$ with $u \neq v$ and take $n \in S^{d-1}$. Then there exist $(\gamma,\xi) \in [0,\pi] \times S^{d-2}$ 
 and a measurable bijective function $\Gamma(u-v, \cdot): S^{d-2} \longrightarrow S^{d-2}(u-v),\ \xi \longmapsto \Gamma(u-v, \xi)$
 such that
  \begin{align}\label{N2}
   n = \cos(\gamma)\frac{u-v}{|u-v|} + \sin(\gamma)\frac{\Gamma(u-v, \xi)}{|u-v|},
  \end{align}
 where $\gamma$ is the angle between $u-v$ and $n$, i.e. it holds that $(u-v, n) = \cos(\gamma)|u-v|$,
\end{Lemma}
\begin{proof}
 Fix $u,v \in \R^{d}$ with $u \neq v$ and $n \in S^{d-1}$. Then $n = n_0 + n_1$, where
 \[
  n_0 = \left( n, \frac{u-v}{|u-v|} \right) \frac{u-v}{|u-v|} = \cos(\gamma)\frac{u-v}{|u-v|}, \qquad
  n_1 = n - \left( n, \frac{u-v}{|u-v|}\right) \frac{u-v}{|u-v|}.
 \]
 Since $|u-v|n_1 \in S^{d-2}(u-v)$, assertion \eqref{N2} is proved, provided we can find a parameterization $\Gamma(u-v, \cdot)$
 such that $n_1 = \sin(\gamma)\frac{\Gamma(u-v, \xi)}{|u-v|}$ for some $\xi \in S^{d-2}$.
 The construction of $\Gamma$ was given in \cite{T78} for $d = 3$, and then generalized in \cite{FM09} for arbitrary dimension $d \geq 3$.
 Below we briefly summarize their construction. Define $e_d = (0,\dots, 0,1)$. If $u - v = |u-v|e_d$ then let $\Gamma(u-v,\xi) = \xi |u-v|$. Otherwise $\Gamma(u-v,\cdot)$ is obtained by doing an axial rotation which overlaps $e_d$ with $u-v$. To obtain this let $R_{u-v}$ be the symmetry with respect to the hyperplane given by 
 \[
  H_{u-v} = \left\{ x \in \R^d \ | \ \left( x, e_d - \frac{u-v}{|u-v|}\right) = 0\right\}^{\perp}.
 \]
 One can show that
 \[
  \Gamma(u-v, \xi) = |u-v|R_{u-v}(\xi_1,\dots, \xi_{d-1},0), \qquad \xi \in S^{d-2}
 \]
 has the desired properties.
\end{proof}
As a consequence we obtain, for all $u,v \in \R^d$ with $u \neq v$, the parameterization
\begin{align}\label{N1}
 (\gamma, \xi) \longmapsto n(u-v, \gamma, \xi) = \cos(\gamma)\frac{u-v}{|u-v|} + \sin(\gamma)\frac{\Gamma(u-v, \xi)}{|u-v|} \in S^{d-1}.
\end{align}
In order to make use of certain symmetries of the collisions (e.g. the Povzner inequalities in Section 6),
it is convenient to also change the angle $\gamma$.
Let $\theta = \theta(n) \in (0,\pi]$ be the angle between $v^{\star} - u^{\star}$ and $v-u$, i.e. one has
\[
 (v - u, v^{\star} - u^{\star}) = \cos(\theta) |v-u| | v^{\star} - u^{\star}|.
\]
Considering the triangle with endpoints $\frac{v+u}{2}, v, v^{\star}$, the corresponding angles have to satisfy 
$2 \gamma + \theta = \pi$, i.e. $\gamma = \frac{\pi}{2} - \frac{\theta}{2}$.
Therefore, we obtain from \eqref{N1}
\begin{align*}
 n = \sin\left( \frac{\theta}{2} \right) \frac{u-v}{|u-v|} + \cos\left( \frac{\theta}{2}\right) \frac{\Gamma(u-v,\xi)}{|u-v|}
\end{align*}
and, in particular, one has
\begin{align}\label{N}
 (u-v,n) = |u-v|\cos(\gamma) = |u-v| \sin\left( \frac{\theta}{2}\right).
\end{align}
Inserting this into \eqref{COLLISION} gives after a short computation
\begin{align}\label{PARA:01}
 \begin{cases} v^{\star} &= v + \alpha(v,u,\theta,\xi)
 \\ u^{\star} &= u - \alpha(v,u,\theta,\xi) \end{cases},
\end{align}
where
\begin{align}\label{FPE:00}
 \alpha(v,u,\theta,\xi) = \sin^2\left( \frac{\theta}{2}\right)(u-v) + \frac{\sin(\theta)}{2}\Gamma(u-v,\xi).
\end{align}
Note that \eqref{PARA:01} remains true also for $v = u$, if we let $\alpha(v,v,\theta,\xi) = 0 $, i.e. set $\Gamma(0,\xi) = 0$ in \eqref{FPE:00}.
In order to prove certain continuity properties for the collision integral (see Section 4), 
we will need to compare these parameterizations for different values of $u,v$.
It was already pointed out by Tanaka that 
$(u,v) \longmapsto (u-v,n)n$ cannot be smooth.
However, Tanaka has shown in \cite[Lemma 3.1]{T78} that if we allow to shift $\xi$ in a suitable way,
then a weaker form of continuity holds. The latter estimate is sufficient for this work.
Below we recall Tanaka's result for arbitrary dimension $d \geq 3$ which is due to \cite{FM09}.
\begin{Lemma}\cite[Lemma 3.1]{FM09}\label{PARAMETARIZATION}
 There exists a measurable map $\xi_0: \R^d \times \R^d \times S^{d-2} \longrightarrow S^{d-2}$ such that for any $X,Y \in \R^d\backslash \{0\}$,
 the map $\xi \longmapsto \xi_0(X,Y,\xi)$ is a bijection with jacobian $1$ from $S^{d-2}$ onto itself, and
 \[
  |\Gamma(X, \xi) - \Gamma(Y, \xi_0(X,Y,\xi))| \leq 3 |X-Y|, \qquad  \xi \in S^{d-2}.
 \]
\end{Lemma}
 With this parameterization we obtain from Lemma \ref{PARAMETARIZATION}, for all $u,v, \widetilde{u}, \widetilde{v} \in \R^d$, all $\theta \in [0,\pi]$ and all $\xi \in S^{d-2}$, we have
 \begin{align}\label{EQ:13}
  |\alpha(v,u,\theta, \xi) - \alpha(\widetilde{v}, \widetilde{u}, \theta, \xi_0(v-u, \widetilde{v} - \widetilde{u}, \xi))| \leq 2 \theta\left( |v - \widetilde{v}| + |u -  \widetilde{u}|\right).
 \end{align}
We work with the parameterization \eqref{PARA:01}, where $\alpha$ is given by \eqref{FPE:00}.

\subsection{Assumptions}
Take $d \geq 3$ and assume that the collision kernel $B$ is given by a function $\sigma \geq 0$ and a $\sigma$-finite measure $Q$ such that
\begin{align}\label{BCOLL}
 B(|v-u|,n)dn \equiv \sigma(|v-u|)Q(d\theta)d\xi, \qquad \kappa := \int \limits_{(0,\pi]}\theta Q(d\theta) < \infty
\end{align}
where $d\xi$ is the Lebesgue surface measure on $S^{d-2}$ (recall \eqref{N}).
 Moreover we assume that there exist $\gamma \in (-1,2]$ and $c_{\sigma} \geq 1$ such that
 \begin{align}\label{EQ:01}
  |\sigma(|z|) - \sigma(|w|)| \leq c_{\sigma}| |z|^{\gamma} - |w|^{\gamma}|, \ \ z,w \in \R^d \backslash \{0\}
 \end{align}
 and
\begin{align}\label{EQ:07}
 \sigma(|z|) \leq c_{\sigma} 
\begin{cases} |z|^{\gamma}, & \text{ if } \gamma \in (-1, 0], \\ (1 + |z|^2)^{\frac{\gamma}{2}}, & \text{ if } \gamma \in (0,2] \end{cases}.
 \end{align}
Finally we assume that $0 \leq \beta \in C_c^1(\R^{d})$ is symmetric and without loss of generality $\beta \leq 1$.
et us remark that this assumptions include 
the cases studied in \eqref{EQ} as long $s > 3$.
\begin{Remark}
 Consider dimension $d = 3$ and let $B$ be given by \eqref{EQ}.
 \begin{enumerate}
  \item[(a)] If $\gamma \in (-1,2]$ and $\nu \in (0,1)$, 
  then letting $\sigma(|v-u|) = |v-u|^{\gamma}$ and $Q(d\theta) = b(\theta)d\theta$,
  we easily find that \eqref{BCOLL} -- \eqref{EQ:07} are satisfied.
  Hence our assumptions include long-range interactions with $s > 3$.
  \item[(b)] In the case of Maxwellian molecules, i.e. $\gamma = 0$ and $\nu = \frac{1}{2}$, one has $\sigma(|v-u|) = 1$ (see \eqref{EQ:01}).
  By inspection of our proofs, we see that all results obtained in this work remain valid also for the case 
  where $\sigma$ is bounded and globally Lipschitz continuous.
 \end{enumerate}
\end{Remark}

\subsection{Measure-solutions for the Enskog equation}

While \eqref{CINT} and hence \eqref{EQ:03} makes sense only for function-valued solutions, our aim is to study existence solutions to the Enskog equation
in the larger space of probability measures. For this purpose we introduce the weak formulation of the Enskog equation, see \cite{ARS17}.
Below we explain this in more detail. 
Let $f_t$ be a sufficiently smooth solution to \eqref{EQ:03}. 
Testing \eqref{EQ:03} against a smooth function $\psi$ and then integrating by parts gives
\begin{align}\label{EQ:102}
 &\ \int \limits_{\R^{2d}}\psi(r,v)\left( \frac{\partial f_t(r,v)}{\partial t} + v \cdot \nabla_r f_t(r,v)\right)drdv
 \\ \notag &= \frac{d}{dt} \int \limits_{\R^{2d}}\psi(r,v) f_t(r,v)dr dv - \int \limits_{\R^{2d}} (v \cdot \nabla_r \psi)(r,v) f_t(r,v)drdv.
\end{align}
The corresponding collision integral $\mathcal{Q}(f_t,f_t)$ is slightly more delicate, see \cite{ARS17} and \cite{T78} for additional details.
In this case one obtains
\begin{align}\label{EQ:101}
 \int \limits_{\R^{2d}}\psi(r,v)\mathcal{Q}(f_t,f_t)(r,v)dr dv = \int \limits_{\R^{4d}}\sigma(|v-u|)\beta(r-q)(\mathcal{L}\psi)(r,v;u)f_t(r,v)f_t(q,u)drdv dq du,
\end{align}
where $\mathcal{L}\psi$ is, for $\psi \in C^1(\R^{2d})$, with $\Xi = (0,\pi] \times S^{d-2}$ defined by
\[
(\mathcal{L}\psi)(r,v;u) = \int \limits_{\Xi}\left( \psi(r,v + \alpha(v,u,\theta,\xi)) - \psi(r,v)\right) Q(d\theta)d\xi.
\]
Now let for $\psi \in C^1(\R^{2d})$
\begin{align*}
 (\mathcal{A}\psi)(r,v;q,u) &= v \cdot(\nabla_r \psi)(r,v) + \sigma(|v-u|)\beta(r-q)(\mathcal{L}\psi)(r,v;u).
\end{align*}
Then combining \eqref{EQ:102} and \eqref{EQ:101} together with the conservation of probability \eqref{CONSERVATION},
shows that $(f_t)_{t \geq 0}$ satisfies the weak formulation of the Enskog equation
\begin{align}\label{EQ:103}
 \frac{d}{dt}\int \limits_{\R^{2d}} \psi(r,v)f_t(r,v)dr dv =  \int \limits_{\R^{4d}} (\mathcal{A}\psi)(r,v;q,u) f_t(r,v)f_t(q,u)drdvdqdu,
\end{align}
\begin{Remark}
 By \eqref{N} and \eqref{FPE:00} one finds 
 \begin{align}\label{PARA:00}
  |\alpha(v,u,\theta,\xi)| = |v-u|\sin\left(\frac{\theta}{2}\right)
 \end{align}
 and hence 
 \begin{align}\label{FPE:02}
  |\psi(r,v + \alpha(v,u,\theta,\xi)) - \psi(r,v)| \leq |v-u|\sin\left(\frac{\theta}{2}\right)\max\limits_{|\zeta| \leq 2 (|v| + |u|)} |\nabla_{\zeta}\psi(r,\zeta)|
 \end{align}
and hence $(\mathcal{L}\psi)(r,v;u)$ is well-defined for all $r,v,u$ and all $\psi \in C^1(\R^{2d})$.
 From this we deduce that for $\psi \in C_b^1(\R^{2d})$
 \begin{align}\label{EXISTENCE:01}
  |\mathcal{A}\psi(r,v;q,u)| \leq \Vert \nabla_{r}\psi \Vert_{\infty}|v| + \kappa \Vert \nabla_v \psi \Vert_{\infty} |v-u|\sigma(|v-u|)|S^{d-2}|.
 \end{align}
 This shows that also $\mathcal{A}\psi$ is well-defined.
\end{Remark}
Having in mind that the solutions of the Boltzmann/Enskog equation 
\eqref{EQ:03} are supposed to be densities, we
note that the weak formulation \eqref{EQ:103} also makes sense for measures.
Below we give the precise definition of a solution measure to the weak formulation of the Enskog equation.
In order to avoid (at the moment) differentiability issues, we give the weak formulation of the Enskog equation in the integral form.
Denote by $\mathcal{P}(\R^d)$ the space of probability measures and let 
\[
 \langle \psi, \mu \rangle = \int \limits_{\R^{2d}}\psi(r,v) d\mu(r,v)
\]
 be the pairing between $\mu \in \mathcal{P}(\R^d)$ and a $\mu$ integrable function $\psi$.
\begin{Definition}
 Let $\mu_0 \in \mathcal{P}(\R^{2d})$.
 A weak solution to the Enskog equation \eqref{EQ:03} is a family $(\mu_t)_{t \geq 0} \subset \mathcal{P}(\R^{2d})$ such that for all $T > 0$
 \begin{align}\label{INTRO:02}
  \int\limits_{0}^{T}\int \limits_{\R^{2d}}|v|^{1+\gamma^+}d\mu_t(r,v)dt < \infty,
 \end{align}
 where $\gamma^+ = \gamma \vee 0$ and for any $\psi \in C_b^1(\R^{2d})$ we have
 \begin{align}\label{FPE:ENSKOG}
  \langle \psi, \mu_t \rangle = \langle \psi, \mu_0 \rangle + \int \limits_0^t \langle \mathcal{A}\psi, \mu_s \otimes \mu_s \rangle ds, \ \ t \geq 0.
 \end{align}
 A weak solution is conservative if it has finite second moments in $v$ and
 \[
  \int \limits_{\R^{2d}}\begin{pmatrix} v \\ |v|^2 \end{pmatrix}d\mu_t(r,v) = \int \limits_{\R^{2d}}\begin{pmatrix} v \\ |v|^2 \end{pmatrix} d\mu_0(r,v), \ \ t \geq 0.
 \]
\end{Definition}
By observing that the total momentum and total kinetic energy is conserved after the transformation \eqref{CONSERVATION1}
and taking $\psi(r,v) = v$ or $\psi(r,v) = |v|^2$ in \eqref{FPE:ENSKOG} one finds that any reasonably regular solution should be conservative.
Since such choices for $\psi$ do not belong to $C_b^1(\R^{2d})$ additional approximation arguments are required, see Theorem \ref{FPE:LEMMA02}
in Section 4.
\begin{Corollary}\label{FPE:COROLLARY00}
 Let $\gamma \in (-1,2]$. Then any weak solution $(\mu_t)_{t \geq 0}$ to the Enskog equation satisfies conservation of momentum, i.e.
  $\int_{\R^{2d}}v d\mu_t(r,v) = \int_{\R^{2d}}v d\mu_{0}(r,v)$, where $t \geq 0$.
 Moreover, if for any $T > 0$ one has $\int_{0}^{T}\int_{\R^{2d}}|v|^{2+\gamma}d\mu_t(r,v)dt < \infty$,
 then $(\mu_t)_{t \geq 0}$ is a conservative weak solution.
\end{Corollary}
 Both assertions follow from Theorem \ref{FPE:LEMMA02} proved in Section 4.
Another consequence of Theorem \ref{FPE:LEMMA02} is that, under suitable moment conditions, we may differentiate \eqref{FPE:ENSKOG}
and hence rewrite it in a differential form.
\begin{Remark}\label{REMARK:10}
 Let $(\mu_t)_{t \geq 0}$ be a weak solution to the Enskog equation. If $\gamma \in (0,2]$ suppose in addition that
 \begin{align}\label{EQ:30}
  \int \limits_{0}^{T}\int \limits_{\R^{2d}}|v|^{1 + 2\gamma}d\mu_t(r,v)dt < \infty, \ \ \forall T > 0.
 \end{align}
 Then, for any  $\psi \in C_b^1(\R^{2d})$, $t \longmapsto \langle \mathcal{A}\psi, \mu_t \otimes \mu_t \rangle$ is continuous and
 $t \longmapsto \langle \psi, \mu_t\rangle$ is continuously differentiable such that
 \begin{align}\label{EQ:106}
  \frac{d}{dt}\langle \psi, \mu_t \rangle = \langle \mathcal{A}\psi, \mu_t \otimes \mu_t \rangle, \ \ t \geq 0.
 \end{align}
\end{Remark}
Finally, let us remark on the choice of test function space used in the Definition of weak solutions to the Enskog equation.
\begin{Remark}
 Let $(\mu_t)_{t \geq 0} \subset \mathcal{P}(\R^{2d})$ satisfy \eqref{INTRO:02}.
 Then \eqref{FPE:ENSKOG} holds for all $\psi \in C_c^1(\R^{2d})$ if and only if \eqref{FPE:ENSKOG} holds for all $\psi \in C_b^1(\R^{2d})$.
\end{Remark}
The last remark can be shown by classical approximation arguments similar to those of Section 4.

\subsection{The Enskog process}
For a given random variable $Z$ we let $\mathcal{L}(Z)$ be the law of $Z$.
Below we provide, in the spirit of Tanaka for the space-homogeneous Boltzmann equation, the definition of an Enskog process associated to the
Enskog equation.
\begin{Definition}\label{DEF:ENSKOG}
 Let $\mu_0 \in \mathcal{P}(\R^{2d})$ be given.
An Enskog process with initial distribution $\mu_0$ consists of the following:
\begin{enumerate}
 \item[(i)] A probability space $(\mathcal{X},d\eta)$ and a Poisson random measure $N$ with compensator 
   \begin{align}\label{UNIQ:00}
    d\widehat{N}(s,\theta,\xi,\eta,z) = ds Q(d\theta)d\xi d\eta dz \quad \text{ on } \R_+\times \Xi \times \mathcal{X} \times \R_+.
   \end{align}
   defined on a stochastic basis $(\Omega, \F, \F_t, \Pr)$ with the usual conditions.
 \item[(ii)] A c\`{a}dl\`{a}g process $(q_t(\eta),u_t(\eta)) \in \R^{2d}$ on $(\mathcal{X},d\eta)$
 and an $\F_t$-adapted c\`{a}dl\`{a}g process $(R_t,V_t)$ on $(\Omega,\F, \F_t, \Pr)$ such that $\mathcal{L}(R_t,V_t) = \mathcal{L}(q_t,u_t)$ for any $t \geq 0$, $\mathcal{L}(R_0, V_0) = \mu_0$,
 \[
  \int \limits_{0}^{T}\E\left( |V_t|^{1 + \gamma^+} \right) dt < \infty, \ \ \forall T > 0,
 \]
 and for $\widehat{\alpha}(v,r,u,q,\theta,\xi,z) = \alpha(v,u,\theta,\xi) \1_{[0,\sigma(|v - u|)\beta(r - q)]}(z)$
   \begin{align}\label{SDE:ENSKOG}
    \begin{cases} R_t &= R_0 + \int \limits_{0}^{t}V_s ds
    \\ V_t &= V_0 + \int \limits_{0}^{t}\int \limits_{\Xi \times \mathcal{X} \times \R_+} \widehat{\alpha}(V_{s-},R_s, u_{s}(\eta),q_s(\eta),\theta,\xi,z)dN(s,\theta,\xi,\eta,z)\end{cases}.
   \end{align}
 \end{enumerate}
\end{Definition}
Remark that the Enskog process given by this definition is a weak solution to \eqref{SDE:ENSKOG}.
By abuse of notation we let $(R,V)$ stand for an Enskog process in the sense of Definition \ref{DEF:ENSKOG} and omit the dependece on the other parameters whenever no confusion may arise.
The particular choice $(\mathcal{X}, d\eta) = ([0,1],dx)$ was used in Tanaka's original work (see \cite{T78, T87}) for the space-homogeneous Boltzmann equation.
The Enskog process for the non-homogeneous case as considered in \eqref{SDE:ENSKOG} was defined on $\mathcal{X} = D(\R_+; \R^{2d})$ (the Skorokhod space) with $d\eta$ being the law of $(R,V)$, see \cite{ARS17}.

Introduce, for $\psi \in C_c^1(\R^{2d})$ and $\nu \in \mathcal{P}(\R^{2d})$ with $\int_{\R^{2d}}|v|^{1+\gamma^+}\nu(dr,dv) < \infty$,
 the Markov operator
 \begin{align}\label{ENSKOG:OP}
  (A(\nu)\psi)(r,v) = \int \limits_{\R^{2d}}(\mathcal{A}\psi)(r,v;q,u)\nu(dq,du).
 \end{align} 
The next remark relates the Enskog equation \eqref{FPE:ENSKOG} with the notion of an Enskog process.
\begin{Remark}
 Let $(R,V)$ be an Enskog process, and define $\mu_t = \mathcal{L}(R_t,V_t)$, $t \geq 0$.
 An application of the It\^{o} formula shows that 
 \[
  \psi(R_t,V_t) - \psi(R_0,V_0) - \int \limits_{0}^{t} (A(\mu_s)\psi)(R_s,V_s)ds, \qquad t \geq 0
 \]
 is a martingale. Hence $(\mu_t)_{t \geq 0}$ is a weak solution to the Enskog equation. 
\end{Remark}
Consequently, constructing an Enskog process gives immediately a weak solution to the Enskog equation.

\section{Discussion of the main result}

\subsection{The result}
In this work we provide, based on a particle approximation of binary collisions, an existence theory for Enskog processes 
(and hence weak solutions to the Enskog equation). 
Namely, let $n \geq 2$ be the number of interacting particles whose coordinates in phase space are given by
$r = (r_1,\dots, r_n) \in \R^{dn}$ for their positions and $v = (v_1,\dots, v_n) \in \R^{dn}$ for their velocities.
A collision of a particle $(r_k,v_k)$ with another particle $(r_j,v_j)$ results in the change of particle configuration
\begin{align}\label{CONF}
 (r,v) \longmapsto (r,v_{kj}), \qquad \text{ with } \qquad v_{kj} = v + (e_k - e_j)\alpha(v_k,v_j,\theta,\xi),
\end{align}
where $k,j = 1,\dots, d$, $e_l = (0_d,\dots, 0_d,1_d,0_d,\dots, 0_d) \in \R^{dn}$ where $0_d$ is the zero matrix and $1_d$ the identity matrix in $\R^{d \times d}$. 
This defines an integral operator on $C_c^1(\R^{2dn})$ via
\[
 (\mathcal{J}_{kj}F)(r,v) = \frac{1}{2}\int \limits_{\Xi}\left( F(r,v_{kj}) - F(r,v) \right)Q(d\theta)d\xi.
\]
The corresponding Markov operator for the whole particle dynamics is given by
\begin{align}\label{EQ:00}
 (LF)(r,v) = \sum \limits_{k=1}^{n}v_k \cdot (\nabla_{r_k}F)(r,v) 
 + \frac{1}{n}\sum \limits_{k,j=1}^{n}\sigma(|v_k - v_j|)\beta(r_k - r_j)(\mathcal{J}_{kj}F)(r,v)
\end{align}
with domain $C_c^{1}(\R^{2dn})$. 
In Section 5, Theorem \ref{IPS:TH00} we will prove that the corresponding martingale problem $(L, C_c^1(\R^{2dn}), \rho)$ has,
for each initial distribution $\rho \in \mathcal{P}(\R^{2dn})$,
a unique solution in the Skorokhod space $D(\R_+;\R^{2dn})$ endowed with the usual Skorokhod topology.
Let $(\mathcal{X}_1^n, \dots, \mathcal{X}_n^n)$ be the corresponding Markov process
with sample paths in the Skorokhod space 
Here we let $\mathcal{X}_k^n$ = $(\mathcal{R}_k^n, \mathcal{V}_k^n)$, where $\mathcal{R}_k^n$ denotes the position
and $\mathcal{V}_k^n$ the velocity of the particle $k \in \{1,\dots, n\}$.

Suppose that at initial time $t = 0$ all particles are independent and identically distributed, i.e.
$\mathcal{X}_k^n(0)$, $k=1,\dots, n$, are independent as random variables and $\mathcal{L}(\mathcal{X}_k^n(0)) = \mu_0 \in \mathcal{P}(\R^{2d})$,
for all $k = 1, \dots, n$. Define the sequence of empirical measures 
\[
 \mu^{(n)} := \frac{1}{n}\sum \limits_{k=1}^{n}\delta_{\mathcal{X}_k^n},
\]
i.e. probability measures over the Skorokhod space $D(\R_+;\R^{2d})$.
We seek to show that $(\mu^{(n)})_{n \geq 2}$ is tight and, moreover, prove that each limit describes the law of 
an Enskog process in the sense of Definition \ref{DEF:ENSKOG}.
Since, as usual, hard and soft potentials require different estimates, we study these seperately.
Let us start with existence in the simpler case of soft potentials.
\begin{Theorem}\label{TH:04}
 Suppose that $\gamma \in (-1,0]$ and let $\mu_0 \in \mathcal{P}(\R^{2d})$ be such that there exists $\e > 0$ with
 \begin{align*}
  \int \limits_{\R^{2d}}\left( |r|^{\e} + |v|^{2+\gamma} \right)d\mu_0(r,v) < \infty.
 \end{align*}
 Then $(\mu^{(n)})_{n \geq 2}$ is tight, and each limit describes the law of an
 Enskog process $(R,V)$ with initial distribution $\mu_0$.
 Moreover, for each $p \geq 2+\gamma$, there exists a constant $C_p > 0$ such that 
 \begin{align*}
  \E \left(\sup \limits_{s \in [0,t]}|V_s|^p\right) \leq \begin{cases} C_p \left( 1 + \int \limits_{\R^{2d}} |v|^p d\mu_0(r,v) \right) + C_p t^{\frac{p}{|\gamma|}}, & \text{ if } \gamma \neq 0 \\ \left( 1 + \int \limits_{\R^{2d}} |v|^p d\mu_0(r,v) \right)e^{C_p t}, & \text{ if } \gamma = 0 \end{cases},
 \end{align*}
 provided the right-hand side is finite. Moreover, if $p \geq 4$ and $\gamma = 0$, then  
 \begin{align*}
  \E \left( |V_t|^p \right) \leq C_p \left( 1 + \int \limits_{\R^{2d}} |v|^p d\mu_0(r,v) \right) t^p,
 \end{align*}
 provided the right-hand side is finite.
\end{Theorem}
Existence of an Enskog process is here obtained under very mild moment assumptions in space and velocity.
In particular, if $\gamma \in (-1,0)$, then $2+\gamma < 2$ and hence the initial condition does not need to have finite (kinetic) energy. 
If $\mu_0$ has finite energy, then the constructed Enskog process has finite second moments and by Remark \ref{FPE:COROLLARY00} satisfies the conservation laws.

For hard potentials the situation is more subtle. We have to distinguish between the critical case $\gamma = 2$ (critical due to the conservation of kinetic energy)
and the simpler case $\gamma \in (0,2)$. In the simpler case we obtain the following.
\begin{Theorem}\label{TH:05}
 Suppose that $\gamma \in (0,2)$ and let $\mu_0 \in \mathcal{P}(\R^{2d})$ be such that there exists $\e > 0$ with
 \begin{align*}
  \int \limits_{\R^{2d}}\left( |r|^{\e} + |v|^{\frac{2}{2-\gamma}\max\{4,1+2\gamma\}} \right)d\mu_0(r,v) < \infty.
 \end{align*}
 Then $(\mu^{(n)})_{n \geq 2}$ is tight, and each limit describes the law of an
 Enskog process $(R,V)$ with initial distribution $\mu_0$.
 Moreover, for any $p \geq 4$, there exists a constant $C_p > 0$ such that
 \begin{align*}
  \E\left( |V_t|^p \right) \leq C_{p}\left( 1 + \int \limits_{\R^{2d}}|v|^{\frac{2p}{2-\gamma}}d\mu_0(r,v) \right)  t^{\frac{2p}{2-\gamma}}, \ \ t \geq 0,
 \end{align*}
 and if $p + \gamma \geq 4$, then also
 \begin{align*}
  \E \left(\sup \limits_{s \in [0,t]} |V_s|^p\right) \leq C_p \left(1 + \int \limits_{\R^{2d}}|v|^{\frac{2p + \gamma}{2-\gamma}}d\mu_0(r,v)\right)\left(1 +  t^{\frac{2p+2}{2-\gamma}}\right), \ \ t \geq 0,
 \end{align*}
 provided the right-hand sides are finite.
\end{Theorem}
Note that in all cases above existence holds solely under some finite moment assumption, i.e. no small density is required.
Results applicable to the case $\beta(x-y) = \delta_\rho(|x-y|)$ are typically obtained for the cut-off case where, in addition,
$\mu_0(dr,dv) = f_0(r,v) dr dv$ and $f_0$ is sufficiently small (see e.g. \cite{TB87}, \cite{A90}, \cite{AC90}). 
The absence of both conditions in this work is, of course, due to the fact that our collision integral is nonsingular in the spatial variables,
i.e. $\beta$ is a smooth function and not a distribution.

In the critical case $\gamma = 2$ we have to assume that the initial particle distribution is much more localized in the velocity variables.
\begin{Theorem}\label{TH:06}
 Suppose that $\gamma = 2$ and let $\mu_0 \in \mathcal{P}(\R^{2d})$ be such that there exists $\e > 0$ and $a > 0$ with
 \begin{align*}
  \int \limits_{\R^{2d}}\left( |r|^{\e} + e^{a |v|^2} \right)d\mu_0(r,v) < \infty.
 \end{align*}
 Then $(\mu^{(n)})_{n \geq 2}$ is tight, and each limit describes the law of an
 Enskog process $(R,V)$ with initial distribution $\mu_0$. Moreover we have for any $p \geq 1$
 \begin{align*}
  \E \left(\sup \limits_{s \in [0,t]}|V_s|^p\right) < \infty , \ \ t \geq 0.
 \end{align*}
\end{Theorem}
We close the presentation of our results with the next remark.
\begin{Remark}
 Let $\mu_0 \in \mathcal{P}(\R^{2d})$ and set $\mu_t = \mathcal{L}(R_t,V_t)$, where $(R_t,V_t)$ is the Enskog process given
 either by Theorem \ref{TH:04}, Theorem \ref{TH:05} or Theorem \ref{TH:06}. Then $\mu_t$ satisfies (at least) the following moment estimates
 \begin{enumerate}
  \item[(a)] If $\gamma \in (-1,0]$, then $\int_{\R^{2d}} |v|^{2+\gamma}\mu_t(dr,dv) < \infty$, for any $t \geq 0$.
  \item[(b)] If $\gamma \in (0,2)$, then $\int_{\R^{2d}} |v|^{\frac{8}{2-\gamma}}\mu_t(dr,dv) < \infty$, for any $t \geq 0$.
  \item[(c)] If $\gamma = 2$, then $\int_{\R^{2d}}|v|^{p}\mu_t(dr,dv) < \infty$ for all $p \geq 1$ and $t \geq 0$.
 \end{enumerate}
 In any case, Remark \ref{REMARK:10} is applicable and hence $(\mu_t)_{t \geq 0}$ is also a weak solution to the Enskog equation
 in the differential form \eqref{EQ:106}.
\end{Remark}
Below we introduce some additional notation and then explain the main ideas in the proofs.

\subsection{Some notation}
For a given Polish space $E$ we let $\mathcal{P}(E)$ stand for the space of Borel probability measures over $E$.
Let $D(\R_+;E)$ be the corresponding Skorokhod space equipped with the usual 
Skorokhod topology and corresponding (right-continuous) filtration as described in \cite{EK86} or \cite{JS03}.

Denote by $C_b(E)$ the Banach space of continuous bounded functions on $E$ and by $B(E)$ the space of bounded measurable functions.
Let $(A, D(A))$ be an (possibly unbounded) operator $A: D(A) \subset C_b(E) \longrightarrow B(E)$ and fix $\rho \in \mathcal{P}(E)$.
A solution to the martingale problem $(A, D(A), \rho)$ is, by definition, given by $\Pr \in \mathcal{P}(D(\R_+;E))$ such that
$\Pr( x(0) \in \cdot ) = \rho$ and, for any $\psi \in D(A)$,
\[
 \psi(x(t)) - \psi(x(0)) - \int \limits_{0}^{t}(A\psi)(x(s))ds, \ \ t \geq 0
\]
is a martingale with respect to $\Pr$ and the filtration generated by the coordinate process $x(t)$ on $D(\R_+;E)$.
Additional references and results are given in \cite{EK86}.
The extension to time-inhomogeneous martingale problems $(A(t), D(A(t)), \rho)$ is obtained by considering space-time $E \times \R_+$.

For $v \in \R^d$ it is convenient to work with $\langle v \rangle := (1 + |v|^2)^{\frac{1}{2}}$ and we frequently use the elementary inequalities
\[
 \langle v + w \rangle \leq \sqrt{2}(\langle v \rangle + \langle w \rangle), \qquad \langle v + w \rangle \leq \sqrt{2}\langle v \rangle \langle w \rangle.
\]
Here and below $C > 0$ denotes a generic constant which may vary from line to line.

\subsection{Main idea of proof}
The classical approach for the construction of solutions to Boltzmann equations is based on entropy dissipation and compactness methods
(see e.g. \cite{CIP94}, \cite{V98}, \cite{LM12} and the references given therein).
At this point one typically assumes that the initial state $\mu_0$ has finite second moments and a density with finite entropy.
The existence of an Enskog process in the case $\gamma = 0$ was shown in \cite{ARS17}.
In this work we propose a purely stochastic approach to the existence theory for \eqref{SDE:ENSKOG}
based on an approximation via a system of interacting particles performing binary collisions.
As a consequence we obtain existence of solutions for a broader class of collision kernels and
initial states (e.g. without finite entropy or finite energy in the case of soft potentials).
Below we summarize the main steps of our proof.

\textbf{Step 1.} Show that the martingale problem with generator $(L, C_c^1(\R^{2dn}))$ is well-posed for any initial distribution
$\rho \in \mathcal{P}(\R^{2dn})$, see Section 5 for additional details.
The proof of this result mainly relies on the use of classical localization arguments for martingale problems, see \cite{EK86}.

\textbf{Step 2.} Prove propagation of moments for the interacting particle system with constants uniformly in $n$.
 This step is studied in Section 6. The general idea is to apply the It\^{o} formula for
\[
 \langle v \rangle_p := \sum\limits _{k=1}^{n}\langle v_k \rangle^p
\]
and after some manipulations use the Gronwall lemma.
In the case of soft potentials, $\gamma \in (-1,0]$, one easily finds by the mean-value theorem, for any $p \geq 1$,
\[
 \E^n\left( \frac{\langle \mathcal{V}^n(t) \rangle_p}{n} \right) 
 \leq \E^n\left( \frac{ \langle \mathcal{V}^n(0)\rangle_p}{n} \right) + 
  C_p \int \limits_{0}^{t}  \E^n\left( \frac{\langle \mathcal{V}^n(s) \rangle_{p+\gamma} }{n} \right) ds
\]
and hence the desired moment estimates follow from the Gronwall lemma (since $\langle v \rangle_{p+\gamma} \leq \langle v \rangle_p$).
Consider the case of hard potentials with $\gamma \in (0,2]$. In this case we use Povzner-type inequalities (see e.g. \cite[Lemma 3.6]{LM12}, \cite{CIP94} or Lemma \ref{FPE:LEMMA05} from Section 6) instead of the mean-value theorem.
If $\beta$ would be strictly bounded away from zero, then we could use the negative terms appearing in the 
Povzner-type inequalities (see Lemma \ref{FPE:LEMMA05}) and prove by similar ideas to \cite{FM09} creation (and propagation) of exponential moments.
Since in this work we suppose that $\beta \in C_c^1(\R^d)$, these ideas do not apply and we only obtain
\[
 \E^n\left( \frac{\langle \mathcal{V}^n(t)\rangle_{2p}}{n} \right) 
 \leq \E^n\left( \frac{\langle \mathcal{V}^n(0)\rangle_{2p}}{n} \right) 
 + C_p \int \limits_{0}^{t} \E^n\left(  \frac{\langle \mathcal{V}^n(s) \rangle_2}{n} \frac{\langle \mathcal{V}^n(s)\rangle_{2p-2+\gamma}}{n} \right) ds.
\]
Due to the additional kinetic energy on the right-hand side we can not directly apply the Gronwall lemma.
Our analysis thus crucially relies on the fact that the conservation laws for the particle system hold pathwise, i.e.
\begin{align}\label{CONS}
 \sum \limits_{k=1}^{n}\mathcal{V}_k^n(t) = \sum \limits_{k=1}^{n}\mathcal{V}_k^n(0), \qquad 
 \sum \limits_{k=1}^{n}|\mathcal{V}_k^n(t)|^2 = \sum \limits_{k=1}^{n}|\mathcal{V}_k^n(0)|^2, \qquad \text{ a.s. }.
\end{align}
These identities imply that $\langle \mathcal{V}^n(t) \rangle_2 = \langle \mathcal{V}^n(0) \rangle$ and hence,
by conditioning on the initial configuration of particles, we may reduce the order of moments on the right-hand side from 
which we deduce the desired moment estimates.
A similar particle system of binary collisions was considered for the space-homogeneous case \cite{HK90} where analogous pathwise conservation laws were used.
The particle system considered in \cite{FM16} and \cite{X16} is not based on binary collisions and hence does not satisfy
conservation laws in a pathwise sense.

Let us give a short argument why \eqref{CONS} is true.
Observe that, by \eqref{CONF}, we have
\begin{align}\label{EST:02}
 \sum \limits_{l=1}^{n} (v_{kj})_l = \sum \limits_{l=1}^{n}v_l,\qquad
 \sum \limits_{l=1}^{n} |(v_{kj})_l|^2 = \sum \limits_{l=1}^{n}|v_l|^2.
\end{align}
If $Q(d\theta)$ is a finite measure, then only finitely many collisions may appear in any finite amount of time and \eqref{CONS} is obvious
(since it holds due to \eqref{EST:02} from collision to collision).
The general case where $Q(d\theta)$ is given as in \eqref{BCOLL} is proved by the It\^{o} formula and a reasonable representation of the process
$(\mathcal{X}_1^n,\dots, \mathcal{X}_n^n)$ in terms of a stochastic equation driven by a Poisson random measure. 

\textbf{Step 3.} By the Aldous criterion, previous moment estimates and a result of Sznitman \cite[Proposition 2.2.(ii)]{S91},
 we deduce that $\mu^{(n)}$ is tight, see Proposition \ref{RELCOMP} in Section 7 for details.
The corresponding moment bounds for the weak limit can be deduced from the moment bounds proved in Step 1 together with convergence of $\mu^{(n)}$ and  Fatou Lemma.

\textbf{Step 4.} By definition of \eqref{ENSKOG:OP} it is not difficult to see that the law $\nu$ of an Enskog process solves the
 (time-inhomogeneous) Martingale problem $(A(\nu_s), C_c^1(\R^{2d}), \mu_0)$ where $\nu_s(\cdot) = \nu( x(s) \in \cdot)$ is its time-marginal.
 Thus, our aim is to prove that any weak limit $\nu \in \mathcal{P}(D(\R_+;\R^{2d}))$ of $\mu^{(n)}$ solves precisely this martingale problem.
 This step is studied in the second part of Section 7. Based on the classical theory of martingale problems, 
 the main obstacle is devoted to the convergence for the corresponding martingale problems (see Theorem \ref{TH:10}). 
 The main technical issue is related to the singularity of $\sigma$ which implies in view of \eqref{EXISTENCE:01} that $A(\nu)$ defined in \eqref{ENSKOG:OP} is not  continuous in $\nu$ w.r.t. weak convergence.
 Thus we have to introduce another approximation $A_R(\nu)$ such that $A_R(\nu)$ is continuous in $\nu$ and then carefully pass to the limit.

\textbf{Step 5.} Using classical results on the representation of solutions to martingale problems by weak solutions to stochastic equations with jumps
(see e.g. \cite[Appendix A]{HK90}) we deduce existence of an Enskog process.
 This step is explained in the last part of Section 7.

\subsection{Structure of the work}
This work is organized as follows.
In Section 4 we study some analytic properties of weak solutions to the Enskog equation.
Section 5 is devoted to the study of the particle system with generator \eqref{EQ:00}.
Corresponding moment estimates are then studied in Section 6 whereas Section 7 is devoted to the convergence of corresponding
martingale problems when $n \to \infty$.
In particular, the proofs of Theorem \ref{TH:04}, Theorem \ref{TH:05} and Theorem \ref{TH:06} are completed in Section 7.

\section{Some analytic properties of the Enskog equation}

\subsection{Continuity properties of the collision operator}
For $q \geq 0$ we introduce the notation and define
\[
 \mathcal{P}_{q}(\R^{2d}) := \left \{ \mu \in \mathcal{P}(\R^{2d}) \ | \ \Vert \mu \Vert_{q} := \int \limits_{\R^{2d}}\langle v \rangle^q d\mu(r,v) < \infty \right\}.
\] 
The following is crucial for the study of corresponding martingale problems.
\begin{Proposition}\label{FPE:LEMMA01}
 The following assertions hold.
 \begin{enumerate}
  \item[(a)] For any $\psi \in C^1(\R^{2d})$ we have $\mathcal{L}\psi \in C(\R^{3d})$ and hence $\mathcal{A}\psi \in C(\R^{4d})$.
  \item[(b)] For $\gamma \in (-1,2]$ there exists a constant $C > 0$ such that for any $\psi \in C_c^1(\R^{2d})$ and $M > 0$ with $\psi(r,v) = 0$ for $|r| > M$ or $|v| > M$ 
  we have for each $0 < \e < \frac{1}{4}$
  \begin{align}\label{EQ:14}
   &\ |\mathcal{A}\psi(r,v;q,u)| 
    \leq \frac{M}{\e} \| \nabla_r \psi \|_{\infty} + C \left( 1 + \left( \frac{M}{\e}\right)^2 \right)^{\frac{1+\gamma}{2}} \| \nabla_v \psi \|_{\infty}\langle u \rangle^{1+\gamma} 
  \end{align}
  Let $\mu \in \mathcal{P}_{1+\gamma}(\R^{2d})$.
  Then $A(\mu): C_c^1(\R^{2d}) \longrightarrow C_0(\R^{2d})$,
  where $C_0(\R^{2d})$ is the Banach space of continuous functions vanishing at infinity
  and $A(\mu)$ was defined in \eqref{ENSKOG:OP}.
  \item[(c)] Let $(\psi_n)_{n \in \N} \subset C^1(\R^{2d})$ be such that 
  $\sup_{n \geq 1}\sup_{|r| + |v| \leq R}|\nabla_{v}\psi_n(r,v)| < \infty$, for all $R > 0$, and $\psi_n \longrightarrow \psi$ pointwise. Then $\mathcal{L}\psi_n \longrightarrow \mathcal{L}\psi$ pointwise.
 \end{enumerate}
\end{Proposition}
\begin{proof}
 \textit{ (a) } Let $(r_n,v_n) \longrightarrow (r,v)$ and $(q_n,u_n) \longrightarrow (q,u)$. For $\e > 0$ take $\delta > 0$ such that
 $\int_{(0,\delta)}\theta Q(d\theta) < \e$.
 Take $R > 0$ such that $|r_n|,|r| \leq R$ and $(|u| + |v|), (|u_n| + |v_n|) \leq \frac{R}{2}$ for all $n \geq 1$. 
 Writing $\alpha = \alpha(v,u,\theta,\xi)$, $\alpha_n = \alpha(v_n,u_n,\theta,\xi)$ we obtain by \eqref{FPE:02}
 \begin{align*}
  &\ |\mathcal{L}\psi(r_n,v_n;u_n) - \mathcal{L}\psi(r,v;u)| \leq |S^{d-2}| R\e \max \limits_{|\zeta|, |r'| \leq R}|\nabla_{\zeta}\psi(r',\zeta)|
  \\ &+ \int \limits_{[\delta,\pi]}\left| \int \limits_{S^{d-2}}\left( \psi(r, v + \alpha) - \psi(r,v) - \psi(r_n, v_n + \alpha_n) + \psi(r_n,v_n)\right)d\xi\right| Q(d\theta)
 \end{align*}
 Set $\alpha'_n = \alpha(v,u,\theta,\xi_0(v_n-u_n,v-u,\xi))$.
 For the second integral we observe that 
 the transformation $\xi_0(X,Y, \cdot)$ in Lemma \ref{PARAMETARIZATION} has Jacobian 1 so that we may insert 
 $\xi_0 = \xi_0(v_n-u_n,v-u,\xi)$
 in the operator $\mathcal {L}$ to find 
 \begin{align*}
  &\ \int \limits_{[\delta,\pi]}\left| \int \limits_{S^{d-2}}\left( \psi(r, v + \alpha) - \psi(r,v) - \psi(r_n, v_n + \alpha_n) + \psi(r_n,v_n)\right)d\xi\right| Q(d\theta)
  \\ &= \int \limits_{[\delta,\pi]}\left| \int \limits_{S^{d-2}}\left( \psi(r, v + \alpha'_n) - \psi(r,v) - \psi(r_n, v_n + \alpha_n) + \psi(r_n,v_n)\right)d\xi\right| Q(d\theta)
  \\ &\leq  2 |S^{d-2}| \int \limits_{[\delta,\pi]}\theta Q(d\theta) \sup \limits_{|\zeta|, |r'| \leq R}\left(|\nabla_{\zeta}\psi(r',\zeta)| + |\nabla_{r'}\psi(r',\zeta)|\right)\left( |v_n -v| + |u_n - u|\right)
  \\ &\ \ \ + 2 |S^{d-2}| Q([\delta,\pi]) \sup \limits_{|\zeta|, |r'| \leq R}\left(|\nabla_{\zeta}\psi(r',\zeta)| + |\nabla_{r'}\psi(r',\zeta)|\right) \left( |v_n -v| + |r_n - r|\right),
 \end{align*}
 where we have also used \eqref{EQ:13}. 
 This proves part (a).
 \\ \textit{(b) } Let $\mu \in \mathcal{P}_{1+\gamma}(\R^{2d})$ and $\psi \in C_c^1(\R^{2d})$.
 Then $\mathcal{A}\psi \in C(\R^{4d})$ and using \eqref{EXISTENCE:01} together with the moment properties of $\mu$
 we conclude  by dominated convergence theorem that $A(\mu)\psi \in C(\R^{2d})$. 
 Let us prove \eqref{EQ:14}. Take $M > 0$ such that $\psi(r,v) = 0$ if $|r| > M$ or $|v| > M$. 
 Then $(\mathcal{A}\psi)(r,v;q,u) = 0$ for $|r| > M$. Consider $|r| \leq M$. For any $\e \in (0,1/4]$ we have for some constant $C > 0$
 \[
  \1_{\{ |v| \leq \e^{-1}M \} } |\mathcal{A}(r,v;q,u)| \leq \frac{M}{\e}  \Vert \nabla_r \psi \Vert_{\infty} + C \| \nabla_v \psi\|_{\infty}\langle u \rangle^{1+\gamma} \left( 1 + \left( \frac{M}{\e} \right)^2 \right)^{ \frac{1+\gamma}{2}}.
 \]
 Secondly we obtain
 \begin{align*}
  \{|v| > \e^{-1}M \} \cap \{|v+\alpha| \leq M \} 
  &\subset \{|v| > \e^{-1}M \} \cap \{|\alpha| \geq |v| - M \} 
  \\ &\subset \{|v| > \e^{-1}M \} \cap \{|\alpha| \geq (1 - \e)|v| \}
  \\ &\subset \left \{\frac{M}{\e} < |v| \leq \frac{|u|}{\sqrt{2}(1 - \e) - 1} \right \},
 \end{align*}
 where $\sqrt{2}(1 - \e) - 1 \geq \sqrt{2}\frac{3}{4} - 1 > 0$ and in the last step we have used $|\alpha(v,u,\theta,\xi)| \leq \frac{|v| + |u|}{\sqrt{2}}$,
 see \eqref{PARA:00}. Then 
 \begin{align*}
   \1_{ \{ |v| > \e^{-1}M \}} |\mathcal{A}\psi(r,v;q,u)| 
  &\leq \1_{ \{ |v| > \e^{-1}M \} } \| \nabla_v \psi \|_{\infty} \int\limits_{ \Xi} |\alpha(v,u,\theta,\xi)| \1_{ \{ |v + \alpha| \leq M \}} \sigma(|v-u|)\beta(r-q)Q(d\theta)d\xi
  \\ &\leq C \| \nabla_v \psi \|_{\infty}  |v-u| \sigma(|v-u|)\1_{ \{ \e^{-1}M < |v| \leq \frac{|u|}{ \sqrt{2}(1 - \e) - 1} \} }
  \\ &\leq C \| \nabla_v \psi \|_{\infty} \1_{\{ (\sqrt{2}(1 - \e) - 1)\e^{-1}M < |u| \} } \langle u \rangle^{1+\gamma}
 \end{align*}
 where in the last inequality we have used $|v-u| \sigma(|v-u|) \leq C \left( \langle v \rangle^{1+\gamma} + \langle u \rangle^{1+\gamma} \right)$
 together with 
 \[
  \1_{ \{ \e^{-1}M < |v| \leq \frac{|u|}{ \sqrt{2}(1 - \e) - 1} \} }  \langle v \rangle^{1 + \gamma} 
 \leq \1_{\{ (\sqrt{2}(1 - \e) - 1)\e^{-1}M < |u| \} }  C \langle u \rangle^{1+\gamma},
 \]
 where the generic constant $C$ is independent of $\e$.
 This proves \eqref{EQ:14}. Integrating above estimates w.r.t. $\mu$ gives
 \[
  \1_{ \{ |v| > \e^{-1}M \} }|(A(\mu)\psi)(r,v)| \leq  C \| \nabla_v \psi \|_{\infty}  \int \limits_{\R^{2d}} \1_{\{ (\sqrt{2}(1 - \e) - 1)\e^{-1}M < |u| \} } \langle u \rangle^{1+\gamma} d\mu(q,u) \longrightarrow 0, \ \ \e \to 0.
 \]
 i.e. $A(\mu)\psi \in C_0(\R^{2d})$.
 Assertion (c) is a consequence of \eqref{FPE:02} and dominated convergence.
\end{proof}
The next statement establishes continuous dependence of $A(\mu)$ on $\mu$.
\begin{Proposition}\label{FPE:PROP00}
 Let $(\mu_n)_{n \geq 1} \subset \mathcal{P}_{1+\gamma}(\R^{2d})$ and take $\mu \in \mathcal{P}_{1+\gamma}(\R^{2d})$.
 Suppose that $(r_n,v_n) \longrightarrow (r,v)$, $\mu_n \longrightarrow \mu$ weakly and $\| \mu_n \|_{1+\gamma} \longrightarrow \| \mu \|_{1+\gamma}$.
 Then
 \[
  \lim \limits_{n \to \infty}(A(\mu_n)\psi)(r_n,v_n) = (A(\mu)\psi)(r,v), \qquad \psi \in C_b^1(\R^{2d}).
 \]
\end{Proposition}
\begin{proof}
 Write 
 \begin{align}\label{EST:00}
 | A(\mu)\psi(r,v) - A(\mu_n)\psi(r_n,v_n)| &\leq |A(\mu)\psi(r,v) - A(\mu_n)\psi(r,v)| + |A(\mu_n)\psi(r,v) - A(\mu_n)\psi(r_n,v_n)|.
 \end{align}
 By Proposition \ref{FPE:LEMMA01}.(a) we have $\mathcal{A}\psi \in C(\R^{4d})$ and by \eqref{EXISTENCE:01} we obtain 
 \[
  |\mathcal{A}\psi(r,v;q,u)| \leq \| \nabla_r \psi \|_{\infty} \langle v \rangle + \kappa c_{\sigma} \| \nabla_v \psi \|_{\infty}\left( \langle v \rangle^{1+\gamma} + \langle u \rangle^{1+\gamma}\right)|S^{d-2}|.
 \]
 Hence by definition of $A(\mu)$ (see \eqref{ENSKOG:OP}) and the hypthesis $\| \mu_n \|_{1+\gamma} \longrightarrow \| \mu \|_{1+\gamma}$, we conclude that the first term in \eqref{EST:00} tends to zero.
 For the second term take $R > 0$ and write
 \begin{align*}
  &\ | A(\mu_n)\psi(r_n,v_n) - A(\mu_n)\psi(r,v)| 
  \\ &\leq \int \limits_{\R^{2d}} |\mathcal{A}\psi(r_n,v_n;q,u) - \mathcal{A}\psi(r,v;q,u)|\1_{\{ |u| > R \}} d\mu_n(q,u)
  \\ &\ \ \ + \int \limits_{\R^{2d}} |\mathcal{A}\psi(r_n,v_n;q,u) - \mathcal{A}\psi(r,v;q,u)|\1_{ \{|u| \leq R \}} d\mu_n(q,u) 
  \\ &\leq C \int \limits_{\R^{2d}} \langle u \rangle^{1+\gamma}\1_{ \{|u| > R \} } d\mu_n(q,u)
   + \sup \limits_{q\in \R^d, \ |u| \leq R}|\mathcal{A}\psi(r_n,v_n,q,u) - \mathcal{A}\psi(r,v,q,u)|,
 \end{align*}
 where $C$ is allowed to depend on $v$ and we have used \eqref{EXISTENCE:01} in the last inequality.
 Since $\| \mu_n \|_{1+\gamma} \longrightarrow \| \mu \|_{1+\gamma}$ it follows that
 \[
 \lim \limits_{R \to \infty} \sup \limits_{n \in \N} \int_{\R^{2d}} \langle u \rangle^{1+\gamma}\1_{ \{|u| > R \} } d\mu_n(q,u) = 0.
 \]
 For the other term we use 
 \begin{align*}
   |\mathcal{A}\psi(r_n,v_n,q,u) - \mathcal{A}\psi(r,v,q,u)| 
   &\leq \left| \sigma(|v_n - u|)(\mathcal{L}\psi)(r_n,v_n;u) - \sigma(|v-u|)(\mathcal{L}\psi)(r,v;u)\right|
  \\ &\ \ \ + \| \nabla \beta\|_{\infty} |r_n -r| \sigma(|v-u|)| (\mathcal{L}\psi)(r,v;u)|
  \\ &\ \ \ + | v_n \cdot (\nabla_r \psi)(r_n, v_n) - v \cdot (\nabla_r \psi)(r,v)|.
 \end{align*}
 All terms tend to zero uniformly in $(q,u) \in \R^d \times \{ |u| \leq R\}$ (see also the proof of Proposition \ref{FPE:LEMMA01}.(a)).
 This proves the assertion. 
\end{proof}

\subsection{Continuity of moments}
Below we provide a sufficient condition under which \eqref{FPE:ENSKOG} can be extended to all $\psi \in C^1(\R^{2d})$ satisfying for some $q \geq 1$
\begin{align}\label{EQ:24}
 \| \psi \|_{C_q^1} := \sup\limits_{(r,v) \in \R^{2d}}|\nabla_r \psi(r,v)|+ \sup \limits_{(r,v) \in \R^{2d}} \frac{|\psi(r,v)| }{1 + |v|^{q}} + \sup \limits_{(r,v) \in \R^{2d}} \frac{|\nabla_v \psi(r,v)|}{1 + |v|^{q-1}} < \infty.
\end{align}
\begin{Theorem}\label{FPE:LEMMA02}
 Let $(\mu_t)_{t \geq 0}$ be a weak solution to the Enskog equation and let $q \geq 1$. Assume that 
 \begin{align}\label{FPE:04}
  \int \limits_{0}^{T}\int \limits_{\R^{2d}}|v|^{q + \gamma^+}d\mu_t(r,v)dt < \infty, \ \ \forall \ T > 0
 \end{align}
 and $\gamma \in (-1,2]$.
 Then for any $\psi \in C^1(\R^{2d})$  satisfying \eqref{EQ:24}, the map $t \longmapsto \langle \mathcal{A}\psi, \mu_t \otimes \mu_t \rangle$ is locally integrable 
 and \eqref{FPE:ENSKOG} holds. Moreover $t \longmapsto \langle \psi, \mu_t \rangle$ is continuous for any $\psi \in C(\R^{2d})$ with
 \begin{align}\label{EQ:08}
  \sup \limits_{(r,v) \in \R^{2d}} \frac{|\psi(r,v)|}{1 + |v|^q} < \infty.
 \end{align}
\end{Theorem}
\begin{proof}
 Let $T > 0$ and fix any $\psi$ as in \eqref{EQ:24}. Then, by \eqref{FPE:02}, we get 
 \begin{align*}
  | \psi(r,v + \alpha(v,u,\theta,\xi)) - \psi(r,v)| &\leq |v-u| \theta \max \limits_{|\zeta| \leq 2(|v|+|u|)} |\nabla_{\zeta} \psi(r,\zeta)|
 \\ &\leq |v-u| \theta \sup \limits_{(r,v) \in \R^{2d}} \frac{|\nabla_v \psi(r,v)|}{1 + |v|^{q-1}} \max \limits_{|\zeta| \leq 2(|v|+|u|)} (1 + |\zeta|^{q-1})
 \\ &\leq C \theta |v-u| \left( \langle v \rangle^{q-1} + \langle u \rangle^{q-1} \right) \sup \limits_{(r,v) \in \R^{2d}} \frac{|\nabla_v \psi(r,v)|}{1 + |v|^{q-1}}.
 \end{align*}
 This implies that
 \begin{align*}
  |(\mathcal{A}\psi)(r,v;q,u)|
  &\leq \| \nabla_r \psi \|_{\infty} |v| + \sigma(|v-u|)\beta(r-q)|(\mathcal{L}\psi)(r,v;u)|
  \\ &\leq \| \nabla_r \psi \|_{\infty} \langle v \rangle + C\left( \langle v\rangle^{1+\gamma} + \langle u \rangle^{1+ \gamma}\right)\left( \langle v \rangle^{q-1} + \langle u \rangle^{q-1} \right)\sup \limits_{(r,v) \in \R^{2d}} \frac{|\nabla_v \psi(r,v)|}{1 + |v|^{q-1}}
  \\ &\leq C \left( \langle v \rangle  +  \langle v \rangle^{q + \gamma} + \langle u \rangle^{q+\gamma} + \langle v\rangle^{1+\gamma}\langle u \rangle^{q-1} + \langle v \rangle^{q - 1} \langle u \rangle^{1+\gamma} \right)\| \psi \|_{C_q^1}
  \\ &\leq C \left( \langle v \rangle^{q + \gamma} + \langle u \rangle^{q+\gamma} \right)\| \psi \|_{C_q^1}
 \end{align*} 
 where in the last inequality we have used  the Young inequality 
 \begin{align}\label{EQ:19}
  \langle v\rangle^{1+\gamma}\langle u \rangle^{q-1} \leq \frac{1+\gamma^+}{q + \gamma} \langle v \rangle^{q+\gamma^+} + \frac{q - 1}{q+ \gamma^+} \langle u \rangle^{q+\gamma^+}.
 \end{align} 
 By \eqref{FPE:04} we see that $t \longmapsto \langle \mathcal{A}\psi, \mu_t \otimes \mu_t \rangle$ is locally integrable.
 Let us prove that \eqref{FPE:ENSKOG} holds for $\psi \in C^1(\R^{2d})$ satisfying \eqref{EQ:24}. 
 Take $g \in C_c^{\infty}(\R_+)$ with $\1_{[0,1]} \leq g \leq \1_{[0,2]}$ and set
 $\psi_n(r,v) = g\left(  \frac{|v|^2}{n^2} \right)\psi(r,v)$. Then $\sup_{n \in \N}  \| \psi_n \|_{C_q^1}< \infty$
 and clearly $\psi_n \longrightarrow \psi$ pointwise. Using equation \eqref{FPE:ENSKOG} with $\psi_n$ in the place of $\psi$, we obtain 
\begin{align}\label{EST:01}
  \langle \psi_n, \mu_t \rangle = \langle \psi_n, \mu_0 \rangle + \int \limits_0^t \langle \mathcal{A}\psi_n, \mu_s \otimes \mu_s \rangle ds, \ \ t \geq 0, \ \ n \in \N
 \end{align}
 and it suffices to show that we can pass to the limit $n \to \infty$.
 Clearly we have $\langle \psi_n, \mu_t \rangle \longrightarrow \langle \psi, \mu_t \rangle$ for any $t \geq 0$.
 For the integral term in \eqref{EST:01} we apply Proposition \ref{FPE:LEMMA01}.(c) so that $\mathcal{A}\psi_n \longrightarrow \mathcal{A}\psi$ pointwise,
 and by the above estimate we obtain
 \[
  |(\mathcal{A}\psi_n)(r,v;q,u)| \leq C\left( \langle v\rangle^{q+\gamma} + \langle u \rangle^{q+ \gamma}\right) \sup \limits_{n \in \N} \| \psi_n \|_{C_q^1} < \infty
 \]
 and in view of  \eqref{FPE:04} we can apply dominated convergence theorem, to pass to the limit $n \to \infty$ in \eqref{EST:01}.

 It remains to prove that $t \longmapsto \langle \psi, \mu_t \rangle$ is continuous for any $\psi \in C(\R^{2d})$ which satisfies \eqref{EQ:08}.
 This property certainly holds for any $\psi_q \in C_b^1(\R^{2d})$ and hence by standard density arguments also for any $\psi \in C_b(\R^{2d})$.
 Next, using \eqref{FPE:ENSKOG} for the particular choice $\psi_q(r,v) := \langle v \rangle^q$, which is possible since $\psi_q$ satisfies \eqref{EQ:24},
 and using that  $t \longmapsto \langle \mathcal{A}\psi_q, \mu_t \otimes \mu_t \rangle$ is locally integrable,
 we see that $t \longmapsto \int_{\R^{2d}}\langle v \rangle^q \mu_t(dr,dv)$ is continuous. This readily implies the assertion.
\end{proof}

\section{The interacting particle system}
In this section we study the particle dynamics given by \eqref{EQ:00} associated with the Enskog process.
The following is the main result for this section.
\begin{Theorem}\label{IPS:TH00}
 For each $\rho \in \mathcal{P}(\R^{2dn})$ there exists a unique solution $\Pr_{\rho} \in \mathcal{P}(D(\R_+;\R^{2dn}))$ to 
 the martingale problem $(L, C_c^1(\R^{2dn}), \rho)$.
 Moreover this solution is the unique weak solution to the stochastic equation
 \begin{align}\label{IPS:EQ01}
  \begin{cases}\mathcal{R}(t) &= \mathcal{R}(0) + \int \limits_{0}^{t}\mathcal{V}(s)ds
  \\ \mathcal{V}(t) &= \mathcal{V}(0) + \int \limits_{0}^{t}\int \limits_{E \times \R_+}G(\mathcal{R}(s),\mathcal{V}(s-),z,\theta,\xi,l,l')dN(s,l,l',\theta,\xi,z) \end{cases}
 \end{align}
 where $G(r,v,z,\theta,\xi,l,l') := (e_l - e_{l'}) \alpha(v_l, v_{l'},\theta, \xi)\1_{[0,\sigma(|v_l - v_l'|)\beta(r_l-r_{l'})]}(z)$
 and $N$ is a Poisson random measure defined on a filtered probability space $(\Omega,\F,\F_t,\Pr)$ with right-continuous filtration 
 and compensator $d\widehat{N} = ds d\nu dz$ on $\R_+ \times E \times \R_+$, $E = \{1,\dots, n\}^2 \times \Xi$,
 \begin{align}\label{EQ:06}
  \nu(dl,dl',d\theta, d\xi) = \frac{1}{2n} \sum \limits_{k=1}^{n}\sum \limits_{j=1}^{n}\delta_k(dl)\delta_j(dl')Q(d\theta)d\xi
 \end{align}
 and $\mathcal{L}(\mathcal{R}(0), \mathcal{V}(0)) = \rho$ such that $\mathcal{L}(\mathcal{R}, \mathcal{V}) = \Pr_{\rho}$.
 This solution satisfies the conservation laws \eqref{CONS}.
\end{Theorem}
The rest of this section is devoted to the proof of Theorem \ref{IPS:TH00}.
\begin{Lemma}\label{TH:02}
 The operator $L$ satisfies the following properties
\begin{enumerate}
  \item[(a)] $LF \in C_c(\R^{2dn})$ for any $F \in C_c^1(\R^{2dn})$.
  \item[(b)] $L$ satisfies the positive maximum principle, i.e. let $F \in C_c^1(\R^{2dn})$ and $(r_0,v_0) \in \R^{2dn}$ be a global maximum of $F$, then $(LF)(r_0,v_0) \leq 0$.
  \item[(c)] $L$ is conservative, i.e. there exists $(F_m)_{m \geq 1} \subset C_c^1(\R^{2dn})$ such that $F_m \longrightarrow 1$ and $LF_m \longrightarrow 0$ bounded pointwise as $m \to \infty$.
\end{enumerate} 
\end{Lemma}
\begin{proof}
  Let us first prove that $\mathcal{J}_{kj}: C_c^1(\R^{2dn}) \longmapsto C_c(\R^{2dn})$, $1 \leq k,j \leq n$.
 Fix any $F \in C_c^1(\R^{2dn})$. Using similar arguments to Proposition \ref{FPE:LEMMA01}.(a) we find that $\mathcal{J}_{kl}F \in C(\R^{2dn})$.
 For convenience of notation we introduce $|r|_2^2 = \sum_{k=1}^{n}|r_k|^2$ and likewise $|v|_2^2$.
 Take $M > 0$ such that $F(r,v) = 0$ whenever $|r|_2^2 \geq M$ or $|v|_2^2 \geq M$. 
 By definition of $\mathcal{J}_{kj}$ we get $(\mathcal{J}_{kj}F)(r,v) = 0$ if $|r|_2^2 \geq M$. 
 Consider $|v|_2^2 \geq M$. Conservation of kinetic energy implies $|v_{kj}|_2^2 = |v|_2^2 \geq M$ and thus $(\mathcal{J}_{kj}F)(r,v) = 0$ (see \eqref{EST:02}).
 This proves $LF \in C_c(\R^{2dn})$.
 Let $(r_0,v_0) \in \R^{2dn}$ be a global maximum of $F$. Then by definition of $\mathcal{J}_{kj}$ we obtain
 $\mathcal{J}_{kj}F(r_0,v_0) \leq 0$ and hence $(LF)(r_0,v_0) \leq 0$, i.e. the positive maximum principle holds.

 Let us prove that $L$ is conservative.
 Take $\psi \in C^{\infty}(\R_+)$ such that $\1_{[0,1]} \leq \psi \leq \1_{[0,2]}$ and set, for $m \in \N$,
 $\psi_m(r) := \psi\left( \frac{|r|_2^2}{m^2}\right)$ and $F_m(r,v) = \psi_m(r)\psi_m(v)$.
 Then $F_m \in C_c^1(\R^{2dn})$ and $F_m \longrightarrow 1$ bounded pointwise as $m \to \infty$. Using again the fact that $|v_{kj}|_2^2 = |v|_2^2$ combined with $\nabla_r \psi_m(r) = 0$ if $|r| \leq m$, we conclude that
 $(LF_m)(r,v) = \sum_{k=1}^{n}v_k  \cdot (\nabla_{r_k} \psi_m)(r) \psi_m(v) \longrightarrow 0$ bounded pointwise as $m \to \infty$.
\end{proof}
Next we consider the case where $\sigma$ is bounded. 
The general statement can be then deduced by suitable localization (see  \cite[Chapter 4, Theorem 6.1, Theorem 6.2]{EK86}).
For $m \geq 1$ let 
\[
 U_m = \begin{cases} \left\{ (r,v) \ | \ \sum_{k=1}^{n}|v_k|^2 < m^2 \right\}, & \text{ if } \gamma \in (0,2]
 \\ \left\{ (r,v)  \ | \ \forall k,j \in \{1,\dots, n\} \text{ s.t. } |v_k - v_j|^{\gamma} < m \right\}, & \text{ if } \gamma \in (-1,0]\end{cases}
\]
and take $g_m \in C_c^1(\R^{2dn})$ such that $g_m \in C^1(\R^{dn})$ with $\1_{U_{m-1}} \leq g_m \leq \1_{U_m}$ and $U_0 := \emptyset$.
Then $g_m \cdot \sigma$ is bounded. Let $L_mF$ be given by $LF$ with $\sigma$ replaced by $g_m \sigma$.
We consider weak solutions $\mathcal{X}^m(t) = (\mathcal{R}^m(t), \mathcal{V}^m(t))$ to
\begin{align}\label{EQ:04}
 \begin{cases} \mathcal{R}^m(t) &= \mathcal{R}(0) + \int \limits_{0}^{t}g_m(\mathcal{V}^m(s))\mathcal{V}^m(s)ds
 \\ \mathcal{V}^m(t) &= \mathcal{V}(0) + \int \limits_{0}^{t}\int \limits_{E \times [0,c_m]}G_m(\mathcal{R}^m(s), \mathcal{V}^m(s-),z,\theta,\xi,l,l') dN(s,l,l',\theta,\xi,z) \end{cases}
\end{align}
where $N$ is a Poisson random measure defined on a stochastic basis $(\Omega, \F, \F_t, \Pr)$ with right-continuous filtration and compensator $ds d\nu dz$ 
on $\R_+ \times E \times [0, c_m]$, where $c_m > 0$ is some sufficiently large constant and
\begin{align}\label{GMDEF}
 G_m(r,v,z,\theta,\xi,l,l') := (e_l - e_{l'}) \alpha(v_l, v_{l'},\theta, \xi)\1_{[0,g_m(v)\sigma(|v_l - v_l'|)\beta(r_l-r_{l'})]}(z).
\end{align}
Using the It\^{o} formula one can show that any weak solution to \eqref{EQ:04} gives a solution to the martingale problem posed by $(L_m, C_c^1(\R^{2dn}))$.
\begin{Proposition}\label{IPS:PROPOSITION02}
 Let $m \geq 1$ be fixed.
 Then for each $\rho \in \mathcal{P}(\R^{2dn})$ there exists a unique solution to the martingale problem $(L_m, C_c^1(\R^{2dn}), \rho)$ in the Skorokhod space.
 Moreover, this solution can be obtained as a weak solution to \eqref{EQ:04}.
 Such a solution satisfies for $t \geq 0$ similar conservation laws to \eqref{CONS}.
\end{Proposition}
\begin{proof}
 Lemma \ref{TH:02} applied to $L_m$ together with \cite[Chapter 4, Theorem 3.8, Theorem 5.4]{EK86} yields existence of solutions to 
 the martingale problem $(L_m, C_c^1(\R^{2dn}),\rho)$, for any $\rho \in \mathcal{P}(\R^{2dn})$.
 It follows from \cite[Theorem 2.3]{KURTZ10} that each solution to the martingale problem $(L_m, C_c^1(\R^{2dn}), \rho)$ can be obtained from a weak solution 
 to \eqref{EQ:04}. Let $(\mathcal{R}^m(t),\mathcal{V}^m(t))$ be any weak solution
 and set $F_0(v) = \sum_{k=1}^{n}v_k$ and $F_1(v) = \sum_{k=1}^{n}|v_k|^2$. Then using the definition of $G_m$ together with \eqref{EST:02} a short computation shows that
 \begin{align}\label{EQ:10}
  F_j(v + G_m(v,r,z,\xi,l,l')) - F_j(v) = 0, \ \ j = 0,1.
 \end{align}
 From the It\^{o} formula we conclude that conservation of momentum and energy holds. 
 
 Let us prove uniqueness. 
 Applying \cite[Theorem 2.1]{BK93} we see that it suffices to prove uniqueness for the martingale problem $(L_m, C_c^1(\R^{2dn}), \delta_{(r_0,v_0)})$, 
 for all $(r_0,v_0) \in \R^{2dn}$.
 Again by \cite[Corollary 2.5]{KURTZ10} we only have to prove uniqueness in law for solutions to \eqref{EQ:04} with initial condition $(r_0, v_0) \in \R^{2dn}$.
 Since $m$ is fixed we let for simplicity of notation $\mathcal{R}(t) = \mathcal{R}^m(t)$ and $\mathcal{V}(t) = \mathcal{V}^m(t)$.
 The proof follows some ideas taken from \cite{F06}, but now applied for an interacting particle system.
 Consider for $k \geq 1$ the stochastic equation
 \begin{align*}
  \begin{cases} \mathcal{R}^k(t) &= r_0 + \int \limits_{0}^{t}\mathcal{V}^k(s)ds
  \\ \mathcal{V}^k(t) &= v_0 + \int \limits_{0}^{t}\int \limits_{E \times [0,c_m]}G_m(\mathcal{R}^k(s), \mathcal{V}^k(s-),z,\theta,\xi^k,l,l')\1_{ \{ \theta > \frac{1}{k} \} } dN(s,l,l',\theta,\xi,z) \end{cases}
 \end{align*}
 with $\xi^k = \xi_0(\mathcal{V}_l(s-) - \mathcal{V}_{l'}(s-), V_l^k(s-) - V_{l'}^k(s-), \xi)$. 
 Since $\nu\left(\{1,\dots, n\}^2 \times (k^{-1},\pi ]\times S^{d-2}\right) < \infty$
 it follows that this equation can be uniquely solved from jump to jump.
 Applying the It\^{o} formula and using again \eqref{EQ:10} it is not difficult to show that $\mathcal{V}^k(t)$ satisfies conservation of momentum and energy.
 Recall that here and below we have let $|v|_2^2 = \sum_{j=1}^{n}|v_j|^2$ for $v \in \R^{dn}$.
 Since $\sup_{t \in [0,T]}|\mathcal{R}^k(t) - \mathcal{R}(t)|_2 \leq T \sup_{t \in [0,T]}|\mathcal{V}^k(t) - \mathcal{V}(t)|_2$ it suffices to prove
 \begin{align}\label{IPS:EQ00}
  \E\left( \sup \limits_{t \in [0,T]} |\mathcal{V}^k(t) - \mathcal{V}(t)|_2 \right) \leq a_k C + C \int \limits_{0}^{T}\E\left( \sup \limits_{s \in [0,t]}| \mathcal{V}^k(s) - \mathcal{V}(s)|_2 \right)dt
 \end{align}
 where $C = C(n,m,r_0,v_0, T) > 0$ is some constant independent of $k$ and $0 \leq a_k \longrightarrow 0$ as $k \to \infty$
 (apply e.g. similar arguments to \cite{F06}).
 
 Set $g_m := g_m(\mathcal{V}(s-)), \sigma = \sigma(|\mathcal{V}_l(s-) - \mathcal{V}_{l'}(s-)|)$, $\beta = \beta(\mathcal{R}_l(s) - \mathcal{R}_{l'}(s))$
 and $g_m', \sigma', \beta'$ with $\mathcal{V}, \mathcal{R}$ replaced by $\mathcal{V}^k, \mathcal{R}^k$. 
 Then by the It\^{o} formula and a simple computation we arrive at
 \[
 \E \left( \sup_{t \in [0,T]} | \mathcal{V}^k(t) - \mathcal{V}(t)|_2 \right) \leq I_1+ I_2+I_3+I_4
 \]
 where
 \begin{align*}
  I_1 &= \int \limits_{0}^{T}\int \limits_{E} \E\left( \left| (e_l - e_{l'})\left(\alpha(\mathcal{V}_l, \mathcal{V}_{l'}, \theta,\xi) - \alpha(\mathcal{V}_l^k, \mathcal{V}_{l'}^k, \theta, \xi^k)\right) \right|_2 g_m\sigma \beta \wedge g_m'\sigma' \beta' \right) d\nu ds,
  \\ I_2 &= \int \limits_{0}^{T}\int \limits_{E} \E\left( |(e_l - e_{l'})\alpha(\mathcal{V}_l, \mathcal{V}_{l'},\theta,\xi)|_2|g_m'\sigma' \beta' - g_m\sigma \beta| \right)d\nu ds,
  \\ I_3 &= \int \limits_{0}^{T}\int \limits_{E} \E\left( |(e_l - e_{l'})\alpha(\mathcal{V}_{l}^k, \mathcal{V}_{l'}^k,\theta, \xi^k)|_2|g_m\sigma \beta - g_m'\sigma' \beta'| \right)d\nu ds,
  \\ I_4 &= \int \limits_{0}^{T} \int \limits_{\{1, \dots, n\}^2 \times (0,\frac{1}{k})\times S^{d-2}} \E \left( |(e_l - e_{l'})\alpha(\mathcal{V}_l^k, \mathcal{V}_{l'}^k, \theta, \xi^k)|_2g_m'\sigma' \beta'\right)d\nu ds.
 \end{align*}
 For $I_1$ we easily obtain $I_1 \leq C \int_{0}^{T} \E( \sup_{s \in [0,t]}|\mathcal{V}^k(s) - \mathcal{V}(s)|_2 ) dt$
 where we have used $\sum_{l=1}^{n}|v_k| \leq |v|_2$ for $v \in \R^{dn}$. Concerning $I_4$ we get by conservation of energy
 \begin{align*}
  I_4 &\leq \frac{C}{n} \sum \limits_{l \neq l'} \int \limits_{0}^{T}\int \limits_{[0, \frac{1}{k}) \times S^{d-2}} \E( \theta |\mathcal{V}_l^k - \mathcal{V}_{l'}^k| )Q(d\theta)d\xi ds
      \leq C a_k \int \limits_{0}^{T} \E( |\mathcal{V}^k(s)|_2) ds = C T |v_0|_2 a_k
 \end{align*}
 with $a_k = \int_{[0, \frac{1}{k}) \times S^{d-2}}\theta Q(d\theta)d\xi$. 
 Concerning $I_2 + I_3$ we estimate
 \begin{align*}
  &\ \left(|(e_l - e_{l'})\alpha(\mathcal{V}_l, \mathcal{V}_{l'},\theta,\xi)| + |(e_l - e_{l'})\alpha(\mathcal{V}_{l}^k, \mathcal{V}_{l'}^k,\theta, \xi^k)|\right)|g_m'\sigma' \beta' - g_m\sigma \beta|
  \\ &\leq C \theta |\mathcal{V}_l - \mathcal{V}_{l'}| \left( g_m \sigma |\beta - \beta'| + \beta' \sigma |g_m - g_m'| + \beta' g_m' |\sigma - \sigma'| \right)
  \\ &\ \ \ + C \theta |\mathcal{V}_l^k - \mathcal{V}_{l'}^k| \left( g_m' \sigma' |\beta - \beta'| + \beta \sigma' |g_m - g_m'| + \beta g_m |\sigma - \sigma'| \right)
  \\ &\leq C \theta \left( |\mathcal{R}_l - \mathcal{R}_{l}^k| + |\mathcal{R}_{l'} - \mathcal{R}_{l'}^k| \right)
         + C \theta \left( |\mathcal{V}_l - \mathcal{V}_{l'}| + |\mathcal{V}_l^k - \mathcal{V}_{l'}^k| \right) |\sigma - \sigma'|
  \\ &\ \ \ + C \theta \left( |\mathcal{V}_l - \mathcal{V}_{l'}| + |\mathcal{V}_l^k - \mathcal{V}_{l'}^k| + |\mathcal{V}_l - \mathcal{V}_{l'}|^{1+ \gamma} + |\mathcal{V}_l^k - \mathcal{V}_{l'}^k|^{1+\gamma}\right)|\mathcal{V} - \mathcal{V}^k|_2,
 \end{align*}
 where we have used $|g_m - g_m'| \leq C(\1_{U_m}(\mathcal{V}) + \1_{U_m}(\mathcal{V}^k)) | \mathcal{V} - \mathcal{V}^k|_2$.
 For the second term we apply
 \begin{align*}
    \left( |v-u| + |v' - u'| \right) \left| \sigma(|v-u|) - \sigma(|v' - u'|) \right| 
   \leq C \left( |v - u|^{\gamma} + |v' - u'|^{\gamma}\right)\left( |v-v'| + |u - u'| \right)
 \end{align*}
 whereas for the third term we use
 \begin{align*}
  &\ |\mathcal{V}_l - \mathcal{V}_{l'}| + |\mathcal{V}_l^k - \mathcal{V}_{l'}^k| + |\mathcal{V}_l - \mathcal{V}_{l'}|^{1+ \gamma} + |\mathcal{V}_l^k - \mathcal{V}_{l'}^k|^{1+\gamma}
  \\ &\leq C( |\mathcal{V}_l| + |\mathcal{V}_{l'}| + |\mathcal{V}_l^k| + |\mathcal{V}_{l'}^k| +|\mathcal{V}_l|^{1+\gamma} + |\mathcal{V}_{l'}|^{1+\gamma} + |\mathcal{V}_l^k|^{1+\gamma} + |\mathcal{V}_{l'}^k|^{1+\gamma} )
  \\ &\leq C( 1 + |\mathcal{V}|_2^2 + |\mathcal{V}^k|_2^2 + |\mathcal{V}|_2^4 + |\mathcal{V}^k|_2^4) \leq C(1 + |v_0|_2^2 + |v_0|_2^4)
 \end{align*}
 to obtain $I_2 + I_3 \leq C \int_{0}^{T} \E \left( \sup_{s \in [0,t]}| \mathcal{V}^k(s) - \mathcal{V}(s)|_2 \right)dt$.
 This completes the proof of \eqref{IPS:EQ00}.
\end{proof}

\section{Moment estimates}

\subsection{The general moment formula}
Let $\rho^{(n)} \in \mathcal{P}(\R^{2dn})$ be the initial distribution. 
Denote by $\Pr_{\rho^{(n)}}$ the unique solution to the corresponding martingale problem $(L, C_c^1(\R^{2dn}), \rho^{(n)})$ and let 
$(\mathcal{R}^n(t),\mathcal{V}^n(t))$ be the weak solutions to \eqref{IPS:EQ01} on $(\Omega_n, \F^n, \F_t^n, \Pr^n)$ with law $\Pr_{\rho^{(n)}}$,
see Theorem \ref{IPS:TH00}.
Expectations w.r.t. $\Pr^n$ we denote by $\E^n$. If the initial condition is deterministic, say $\rho^{(n)} = \delta_{(r_0,v_0)}$ for $(r_0,v_0) \in \R^{2dn}$,
we also write $\Pr_{(r_0,v_0)}$ and $\E^n_{(r_0,v_0)}$ to indicate the dependence on the initial condition.
Finally let $\mathcal{X}^n = (\mathcal{X}_1^n, \dots, \mathcal{X}_n^n)$ where $\mathcal{X}_k^n = (\mathcal{R}_k^n, \mathcal{V}_k^n)$.
\begin{Lemma}\label{EXCH}
  Suppose that $\rho^{(n)} \in \mathcal{P}(\R^{2dn})$ is symmetric.
  Then $\mathcal{X}_1^n, \dots, \mathcal{X}_n^n$ are exchangeable as elements in $D(\R_+;\R^{2d})$.
  In particular $\mathcal{X}_k^n$, $k = 1, \dots, n$ are identically distributed.
\end{Lemma}
\begin{proof}
 This follows from the fact that $L$ maps symmetric functions onto symmetric functions and that the martingale problem $(L, C_c^1(\R^{2dn}), \rho^{(n)})$
 is well-posed.
\end{proof}
For $g \in C^1(\R^d)$ define $\langle \cdot \rangle_g: \R^{dn} \longrightarrow \R_+$ via $\langle v \rangle_g =  \sum_{k=1}^{n}g(v_k)$.
Then for $v_{k,j}^{\star} = v_k + \alpha(v_k,v_j,\theta,\xi)$ we obtain for the action of the Markov generator $L$ on $\langle \cdot \rangle_g$
 \begin{align}\label{EST:06}
 (L \langle \cdot \rangle_g)(r,v) &= \frac{1}{n} \sum \limits_{k,j = 1}^{n} \sigma(|v_k - v_j|) \beta(r_k - r_j) \int \limits_{\Xi}\left( g( v_{k,j}^{\star} ) + g( v_{j,k}^{\star}) - g( v_k ) - g( v_j \right) Q(d\theta)d\xi.
\end{align}
The next lemma is a simple consequence of the integration theory for Poisson random measures.
\begin{Lemma}\label{EST:03}
  Let $g \in C^1(\R^d)$ and $(\mathcal{R}^n(t), \mathcal{V}^n(t))$ be a weak solution to \eqref{IPS:EQ01} with 
  symmetric initial distribution $\rho^{(n)} \in \mathcal{P}(\R^{2dn})$. Then
 \begin{enumerate}
  \item[(a)] Set $\tau_M = \inf\{ t > 0 \ | \ \langle \mathcal{V}^n(t) \rangle \geq M \text{ or } \langle \mathcal{V}^n(t-) \rangle \geq M \}$, $M \geq 1$, then
  \begin{align*}
  \E^n( \langle \mathcal{V}^n(t \wedge \tau_M) \rangle_g ) &= \E^n(  \langle \mathcal{V}^n(0) \rangle_g ) + \E^n \left( \int \limits_{0}^{t \wedge \tau_M} (L\langle \cdot \rangle_g)(\mathcal{R}^n(s), \mathcal{V}^n(s))ds \right)
  \end{align*}
 \item[(b)] It holds that
 \[
 \E^n \left( \sup \limits_{ s\in [0,t]} g(\mathcal{V}_1^n(s)) \right) \leq \E^n( g(\mathcal{V}_1^n(0) ) ) + \E^n \left( \int \limits_{0}^{t} \mathcal{N}_g( \mathcal{R}^n(s), \mathcal{V}^n(s)) ds \right)
 \]
 where 
 \begin{align*}
  \mathcal{N}_g(r,v) &= \frac{1}{2n} \sum \limits_{k=1}^{n} \sigma(|v_1 - v_k|) \int \limits_{\Xi} | g(v_1 + \alpha(v_1,v_k, \theta,\xi)) - g(v_1)| Q(d\theta)d\xi
 \\ &\ \ \ + \frac{1}{2n} \sum \limits_{k=1}^{n} \sigma(|v_1 - v_k|) \int \limits_{\Xi} | g(v_1 - \alpha(v_1,v_k, \theta,\xi)) - g(v_1)| Q(d\theta)d\xi.
 \end{align*}
 \end{enumerate}
\end{Lemma}
\begin{proof}
 We obtain, for each $M \geq 1$, with $dN =  dN(s,l,l',\theta,\xi,z)$
 \begin{align}\label{EQ:100}
   \langle \mathcal{V}^n(t \wedge \tau_M) \rangle_g 
 =   \langle \mathcal{V}^n(0) \rangle_g + \int \limits_{0}^{t \wedge \tau_M} \int \limits_{E \times \R_+} \left[ \langle  \mathcal{V}^n(s-) + G(\mathcal{R}^n(s), \mathcal{V}^n(s-), z,\theta,\xi,l,l') \rangle_g -  \langle \mathcal{V}^n(s-) \rangle_g \right] dN,
 \end{align}
 where the stochastic integral is defined pathwise and $G$ is given as in \eqref{IPS:EQ01}. Indeed, it holds that
 \begin{align*}
  &\ \E^n\left( \int \limits_{0}^{t \wedge \tau_M} \int \limits_{E \times \R_+}\left| \langle  \mathcal{V}^n(s) + G(\mathcal{R}^n(s), \mathcal{V}^n(s), z,\theta,\xi,l,l') \rangle_g -  \langle \mathcal{V}^n(s) \rangle_g \right| ds d\nu dz \right)
  \\ &\leq \frac{1}{n} \sum \limits_{k = 1}^{n} \E^n\left( \int \limits_{0}^{t \wedge \tau_M} \int \limits_{E \times \R_+} | g(\mathcal{V}_k^n(s) + G_k(\mathcal{R}^n(s), \mathcal{V}^n(s), z,\theta,\xi,l,l')) - g(\mathcal{V}_k^n(s))| ds d\nu dz \right)
  \\ &\leq \frac{C}{n} \sum \limits_{k=1}^{n} \E^n\left( \int \limits_{0}^{t \wedge \tau_M} \int \limits_{E \times \R_+} |G_k(\mathcal{R}^n(s), \mathcal{V}^n(s), z,\theta,\xi,l,l') | ds d\nu dz \right)
  \\ &\leq \frac{C}{n^2} \sum \limits_{k \neq j} \E^n \left( \int \limits_{0}^{t \wedge \tau_M} \int \limits_{\Xi} |\alpha(\mathcal{V}_k^n(s), \mathcal{V}_j^n(s), \theta,\xi) \sigma(|\mathcal{V}_k^n(s) - \mathcal{V}_j^n(s)|) Q(d\theta)d\xi ds \right) 
  \\ &\leq \frac{C}{n^2} \sum \limits_{k \neq j} \E^n \left( \int \limits_{0}^{t \wedge \tau_M} \left( \langle \mathcal{V}_k^n(s)\rangle^{1 + \gamma} + \langle \mathcal{V}_j^n(s) \rangle^{1 + \gamma} \right) ds \right) < \infty,
 \end{align*}
 where in the second inequality we have used that $\mathcal{V}_k(s)$ and $\mathcal{V}_k^n(s) + G_k(\mathcal{R}^n(s), \mathcal{V}^n(s), z,\theta,\xi,l,l')$
 are both bounded a.s. for $s \in [0,t \wedge \tau_M]$ and $G_k$ is defined in \eqref{GMDEF}. Using the definition of $\tau_M$, we find that the right-hand side is finite, i.e. we have shown \eqref{EQ:100}. Taking expectations in \eqref{EQ:100} gives 
 \begin{align*}
   &\ \E^n( \langle \mathcal{V}^n(t \wedge \tau_M) \rangle_g ) - \E^n(  \langle \mathcal{V}^n(0) \rangle_g )
 \\ &=  \frac{1}{n}\sum \limits_{k=1}^n \E^n\left( \int \limits_{0}^{t \wedge \tau_M} \int \limits_{E \times \R_+} ( g( \mathcal{V}_k^n + G_k) - g(\mathcal{V}_k^n) ) ds d\nu dz \right)
 \\ &=  \frac{1}{2n}\sum \limits_{k,j=1}^n \E^n\left( \int \limits_{0}^{t \wedge \tau_M} \int \limits_{E \times \R_+} \left[ 
 g( \mathcal{V}_k^n + \alpha(\mathcal{V}_k^n, \mathcal{V}_j^n, \theta,\xi) )  - g(\mathcal{V}_k^n) ) \right] \sigma(|\mathcal{V}_k^n - \mathcal{V}_j^n|) \beta(\mathcal{R}_k^n - \mathcal{R}_j^n) ds \right)
 \\ &\ \ \ + \frac{1}{2n}\sum \limits_{k,j=1}^n \E^n\left( \int \limits_{0}^{t \wedge \tau_M} \int \limits_{E \times \R_+} \left[ 
 g( \mathcal{V}_k^n - \alpha(\mathcal{V}_k^n, \mathcal{V}_j^n, \theta,\xi) )  - g(\mathcal{V}_j^n) \right] \sigma(|\mathcal{V}_k^n - \mathcal{V}_j^n|) \beta(\mathcal{R}_k^n - \mathcal{R}_j^n) ds \right)
 \\ &= \E^n \left( \int \limits_{0}^{t \wedge \tau_M} (L \langle \cdot \rangle_g)(\mathcal{R}^n(s), \mathcal{V}^n(s))ds \right).
 \end{align*}
 For the second part we use the same localization argument to obtain a.s.
 \begin{align*}
  g(\mathcal{V}_1^n(t)) &= g(\mathcal{V}_1^n(0)) + \int \limits_{0}^{t} \int \limits_{E \times \R_+}\left[ g( \mathcal{V}_1^n(s) + G_1(\mathcal{R}^n(s), \mathcal{V}^n(s), z,\theta,\xi,l,l')) - g(\mathcal{V}_1^n(s-)) \right] dN
 \end{align*}
 where the stochastic integral is defined pathwise. Taking the supremum on both sides gives
 \begin{align*}
  \sup \limits_{s \in [0,t]} g(\mathcal{V}_1^n(s)) 
  &\leq g(\mathcal{V}_1^n(0)) + \int \limits_{0}^{t} \int \limits_{E \times \R_+}\left| g( \mathcal{V}_1^n(s) + G_1(\mathcal{R}^n(s), \mathcal{V}^n(s), z,\theta,\xi,l,l')) - g(\mathcal{V}_1^n(s-)) \right| dN.
 \end{align*}
 Now taking expectations on both sides yields
 \begin{align*}
  &\ \E^n\left( \sup \limits_{s \in [0,t]} g(\mathcal{V}_1^n(s)) \right)  
 \\ &\leq \E^n \left( g(\mathcal{V}_1^n(0)) \right) 
  + \int \limits_{0}^{t} \int \limits_{E \times \R_+}\E^n \left( \left| g( \mathcal{V}_1^n(s) + G_1(\mathcal{R}^n(s), \mathcal{V}^n(s), z,\theta,\xi,l,l')) - g(\mathcal{V}_1^n(s-)) \right|  \right) ds d\nu dz 
 \\ &\leq  \frac{1}{2n} \sum \limits_{k=1}^{n} \int \limits_{0}^{t} \int \limits_{\Xi} \E^n \left( \left| g( \mathcal{V}_1^n(s) + \alpha(\mathcal{V}_1^n(s), \mathcal{V}_k^n(s),\theta,\xi)) - g(\mathcal{V}_1^n(s)) \right| \sigma(|\mathcal{V}_1^n(s) - \mathcal{V}_k^n(s)| ) \right) Q(d\theta)d\xi ds
 \\ &\ \ \ + \frac{1}{2n} \sum \limits_{k=1}^{n} \int \limits_{0}^{t} \int \limits_{\Xi} \E^n \left( \left| g( \mathcal{V}_1^n(s) - \alpha(\mathcal{V}_1^n(s), \mathcal{V}_k^n(s),\theta,\xi)) - g(\mathcal{V}_1^n(s)) \right| \sigma(|\mathcal{V}_1^n(s) - \mathcal{V}_k^n(s)| ) \right) Q(d\theta)d\xi ds.
 \end{align*}
\end{proof}
Below we apply this result for the particular case $g(v) = \langle v \rangle^p$ where, for simplicity of notation,
we let $\langle v \rangle_p := \langle v \rangle_g = \sum_{k=1}^{n}\langle v_k \rangle^p$.

\subsection{Soft potentials $\gamma \in (-1,0]$}
Here and below suppose that $\gamma \in (-1,0]$.
Let us begin with a Lyapunov-type estimate on the interacting particle system.
\begin{Lemma}
 For any $p \geq 1$ we can find a constant $C = C_p > 0$ such that for $g(v) = \langle v \rangle^p$
 \begin{align}\label{IPS:EQ04}
  \mathcal{N}_g(r,v) \leq C \frac{\langle v \rangle_{p+\gamma}}{n}
 \end{align}
\end{Lemma}
\begin{proof}
 For any $p \geq 1$ there exists a constant $C_p > 0$ such that
\begin{align*}
   \int \limits_{S^{d-2}}| \langle v \pm \alpha(v,u,\theta,\xi)\rangle^p - \langle v \rangle^p| d\xi \leq C_p \theta |v-u| \left( \langle v \rangle^{p-1} + \langle u \rangle^{p-1} \right).
  \end{align*}
 Hence we obtain from \eqref{EQ:19} 
 \begin{align*}
  \mathcal{N}_g(r,v) \leq \frac{C}{n}\sum \limits_{k=1}^{n}|v_1- v_k|^{1+\gamma}\left( \langle v_1 \rangle^{p-1} + \langle v_k \rangle^{p-1} \right)
  \leq C \frac{\langle v \rangle_{p+\gamma}}{n}.
 \end{align*}
\end{proof}
Moment estimates uniformly in $n$ are given below.
\begin{Proposition}\label{IPS:TH01SP}
 Suppose that $\gamma \in (-1,0]$.
 Let $(\mathcal{R}^n(t), \mathcal{V}^n(t))$ be a weak solution to \eqref{IPS:EQ01} with 
 symmetric initial distribution $\rho^{(n)} \in \mathcal{P}(\R^{2dn})$.
 Then for each $p \geq 1$ there exists a constant $C_p > 0$ such that 
  \[
   \E^n\left( \sup \limits_{s \in [0,t]}\langle \mathcal{V}_1^n(s)\rangle^{p}\right) \leq \begin{cases} C_p\E^n\left( \langle \mathcal{V}_1^n(0)\rangle^p \right) + C_p t^{\frac{p}{|\gamma|}}, & \gamma \in (-1,0) \\ \ \\
 \E^n\left( \langle \mathcal{V}_1^n(0)\rangle^p \right)e^{C_pt}, & \gamma = 0 \end{cases}, \ \ t \geq 0,
  \]
 provided the right-hand side is finite.
\end{Proposition}
\begin{proof}
 For $N \geq 1$ let $\tau_N = \inf\{ t > 0 \ | \ \langle \mathcal{V}^n(t) \rangle \geq N \text{ or } \langle \mathcal{V}^n(t-) \rangle \geq N \}$.
 Then by Lemma \ref{EST:03}.(b) and \eqref{EST:06}
 \begin{align*}
  \E^n\left( \frac{\langle \mathcal{V}^n(t \wedge \tau_N) \rangle_p}{n} \right) 
&\leq \E^n\left( \frac{\langle \mathcal{V}^n(0)\rangle_p}{n}\right)
  + C_p \int \limits_0^t \E^n \left(\frac{\langle \mathcal{V}^n(s \wedge \tau_N)\rangle_{p+\gamma}}{n} \right) ds
  \\ &\leq \E^n\left( \frac{\langle \mathcal{V}^n(0)\rangle_p}{n}\right)
  + C_p \int \limits_0^t \E^n \left(\frac{\langle \mathcal{V}^n(s \wedge \tau_N)\rangle_{p}}{n} \right) ds
 \end{align*}
 and by the Gronwall lemma $\E^n\left( \frac{\langle \mathcal{V}^n(t \wedge \tau_N)  \rangle_p }{n} \right) \leq \E^n\left( \frac{\langle \mathcal{V}^n(0) \rangle_p }{n} \right)e^{C_pt} < \infty$. Taking the limit $N \to \infty$ yields 
 $\E^n\left( \frac{\langle \mathcal{V}^n(t)\rangle_{p}}{n} \right) \leq \E^n\left( \frac{\langle \mathcal{V}^n(0)\rangle_{p} }{n}\right)e^{C_p t} < \infty$.
 Again by Lemma \ref{EST:03}.(a) and \eqref{IPS:EQ04}
 \begin{align*}
   \E^n\left( \sup \limits_{s \in [0,t]}\langle\mathcal{V}_1^n(s)\rangle^{p} \right) &\leq \E^n\left( \frac{\langle \mathcal{V}^n(0)\rangle_{p}}{n}\right) + C \int \limits_{0}^{t}\E^n\left( \frac{\langle \mathcal{V}^n(s )\rangle_{p+\gamma}}{n} \  \right)ds
 \end{align*}
 where the right-hand side is finite due to $\langle \mathcal{V}^n(s )\rangle_{p+\gamma} \leq \langle \mathcal{V}^n(s )\rangle_{p}$ and previous estimates.
 Since $\mathcal{X}^n_k,\ k=1,\dots,n$ are identically distributed we get 
 \begin{align}\label{INEQUALITY}
  \notag \E^n\left( \frac{\langle \mathcal{V}^n(s )\rangle_{p+\gamma}}{n}   \right)
  &=  \E^n\left( \langle \mathcal{V}_1^n(s )\rangle^{p+\gamma}  \right)
  \\ &\leq  \E^n\left(  \sup \limits_{r \in [0,s]} \langle \mathcal{V}_1^n(r )\rangle^{p+\gamma}  \right)
  \leq \begin{cases}  \E^n\left(  \sup \limits_{r \in [0,s]} \langle \mathcal{V}_1^n(r )\rangle^{p}  \right)^{1 - \frac{|\gamma|}{p}}, & \gamma \neq 0
 \\  \E^n\left(  \sup \limits_{r \in [0,s]} \langle \mathcal{V}_1^n(r )\rangle^{p}  \right), & \gamma = 0 \end{cases}.
 \end{align}
 For $\gamma = 0$ we deduce the assertion from the Gronwall lemma. 
 For $\gamma \in (-1,0)$ we appply the Bihari-LaSalle inequality (see Lemma \ref{LASALLE}).
\end{proof}

\begin{Remark}\label{REMARK}
 In \eqref{INEQUALITY}, for $\gamma \neq 0$, we may bound
 \[
 \E^n\left(  \sup \limits_{r \in [0,s]} \langle \mathcal{V}_1^n(r )\rangle^{p+\gamma}  \right)
 \leq  \E^n\left(  \sup \limits_{r \in [0,s]} \langle \mathcal{V}_1^n(r )\rangle^{p}  \right)
 \]
 as well. Hence, a weaker conclusion results, namely
 \[
  \E^n\left(  \sup \limits_{r \in [0,s]} \langle \mathcal{V}_1^n(r )\rangle^{p}  \right) \leq \E^n\left(  \langle \mathcal{V}_1^n(0 )\rangle^{p}  \right) e^{C_p t}
 \]
 for all $\gamma \in (-1,0]$ and $p \geq 1$.
\end{Remark}

\subsection{Hard potentials $\gamma \in (0,2]$}
In this part we suppose that $\gamma \in (0,2]$.
The following Povzner-type inequality is basically contained in \cite[Lemma 3.6]{LM12}.
However, we suppose that $p \geq 2$ and hence get a less sharp estimate. The proof works in exactly the same way.
\begin{Lemma}\label{FPE:LEMMA05}
 For all $\theta \in (0,\pi]$ and $p \geq 2$
  \begin{align*}
   &\ \int \limits_{S^{d-2}}\left( \langle v + \alpha(v,u,\theta,\xi) \rangle^{2p} + \langle u - \alpha(v,u,\theta,\xi)\rangle^{2p} - \langle v \rangle^{2p} - \langle u \rangle^{2p} \right)d\xi
   \\ &\leq - \frac{\sin^2(\theta)}{2}\left( \langle v \rangle^{2p} + \langle u \rangle^{2p}\right) 
   + C_p \sin^2(\theta) \sum \limits_{k=1}^{\lfloor \frac{p+1}{2} \rfloor} \left( \langle v \rangle^{2k}\langle u \rangle^{2p-2k} + \langle v \rangle^{2p - 2k}\langle u \rangle^{2k}\right),
  \end{align*}
  where $\lfloor x \rfloor \in \Z$ is defined by $\lfloor x \rfloor \leq x < \lfloor x \rfloor +1$ and $C_p > 0$ is some constant.
\end{Lemma}
We may now deduce from the above a similar Lyapunov-type estimate.
\begin{Lemma}
 The following assertions hold.
 \begin{enumerate}
  \item[(a)] For $\gamma \in [0,2]$ and any $2p \geq 4$ we find a constant $C_p > 0$ such that
 \begin{align}\label{IPS:EQ02}
  \frac{1}{n^2}\sum \limits_{k,j=1}^{n}\sigma(|v_k-v_j|)\int \limits_{\Xi}\left( \langle v_k^{\star} \rangle^{2p} + \langle v_j^{\star} \rangle^{2p} - \langle v_k \rangle^{2p} - \langle v_j \rangle^{2p} \right)Q(d\theta)d\xi
  \leq C_p \frac{\langle v \rangle_2}{n} \frac{\langle v \rangle_{2p-2 + \gamma}}{n}
 \end{align}
 where $v_k^{\star} = v_k + \alpha(v_k,v_j,\theta,\xi)$, $v_j^{\star} = v_j - \alpha(v_k,v_j,\theta,\xi)$ 
  \item[(b)] For $\gamma \in [0,2]$ and any $p \geq 1$ we find a constant $C_p > 0$ such that for $g(v) = \langle v \rangle^p$
  \begin{align}\label{IPS:EQ03}
    \mathcal{N}_g(r,v) \leq C_p \frac{\langle v \rangle_{2p + \gamma}}{n}.
  \end{align}
 \end{enumerate}
\end{Lemma}
\begin{proof}
 \textit{(a)} We get by Lemma \ref{FPE:LEMMA05} for some generic constant $C > 0$ and $k_p = \lfloor \frac{p+1}{2} \rfloor$
 \begin{align*}
  &\ \frac{1}{n^2}\sum \limits_{k,j=1}^{n}\sigma(|v_k-v_j|)\int \limits_{\Xi}\left( \langle v_k^{\star} \rangle^{2p} + \langle v_j^{\star} \rangle^{2p} - \langle v_k \rangle^{2p} - \langle v_j \rangle^{2p} \right)Q(d\theta)d\xi
  \\ &\leq \frac{C}{n^2}\sum \limits_{k,j=1}^{n}\sum \limits_{l=1}^{k_p}(\langle v_k \rangle^{\gamma} + \langle v_j \rangle^{\gamma}) \left( \langle v_k \rangle^{2l} \langle v_j \rangle^{2p - 2l} + \langle v_k \rangle^{2p-2l} \langle v_j \rangle^{2l} \right)
  \\ &\leq \frac{C}{n^2}\sum \limits_{l=1}^{k_p}\left( \langle v \rangle_{2l+\gamma}\langle v \rangle_{2p-2l} + \langle v \rangle_{2p-2l+\gamma} \langle v \rangle_{2l} \right).
 \end{align*}
 Next following the proof of \cite[Lemma 3.7]{LM12} we use $\langle v \rangle_r \leq \langle v \rangle_2^{\frac{s-r}{s-2}}\langle v \rangle_s^{\frac{r-2}{s-2}}$
 for $2 \leq r \leq s$ to deduce, for $2 \leq s_1, s_2 \leq 2p-2+\gamma$,
 \begin{align*}
  \langle v \rangle_{s_1} \cdot \langle v \rangle_{s_2} \leq \langle v \rangle_2^{a}\cdot \langle v \rangle_{2p}^{b} \leq \langle v \rangle_2 \cdot \langle v \rangle_{2p-2 + \gamma},
 \end{align*}
 where we have set
 \[
  a = \frac{2(2p - 2 + \gamma) - (s_1 + s_2)}{2p-2 + \gamma - 2}, \qquad b = \frac{s_1 + s_2 - 2}{2p-2+\gamma - 2}.
 \]
 Estimating each term in the sum proves \eqref{IPS:EQ02}.
 \\ \textit{(b)} The assertion follows exactly by the same arguments as \eqref{IPS:EQ04}.
\end{proof}
Corresponding moment estimates are given below.
\begin{Proposition}\label{IPS:TH01HP}
 Let $(\mathcal{R}^n(t), \mathcal{V}^n(t))$ be a weak solution to \eqref{IPS:EQ01} with initial distribution $\rho^{(n)}$.
 Then the following assertions hold.
 \begin{enumerate}
  \item[(a)] Suppose that $\gamma \in [0,2)$. Then for any $2p \geq 4$ there exists a constant $C_p> 0$ such that
  \[
   \E^n\left( \frac{\langle \mathcal{V}^n(t)\rangle_{2p}}{n}\right) \leq C_p \E^n\left( \frac{\langle \mathcal{V}^n(0)\rangle_{2p}}{n}\right) + C_p \E^n\left( \frac{\langle \mathcal{V}^n(0)\rangle_{\frac{4p}{2-\gamma}}}{n}\right) t^{\frac{2p}{2-\gamma}},
  \]
  provided the right-hand side is finite.
  \item[(b)] Suppose that $\gamma = 2$. Then for any $p \geq 2$ there exists $C_p > 0$ such that
  \[
   \E^n\left( \frac{\langle \mathcal{V}^n(t)\rangle_{2p}}{n}\right) \leq \E^n\left( \frac{\langle \mathcal{V}^n(0)\rangle_{2p}}{n} e^{C_p\frac{\langle \mathcal{V}^n(0)\rangle_2}{n}t}\right), \ \ t \geq 0,
  \]
  provided the right-hand side is finite.
 \end{enumerate}
\end{Proposition}
\begin{proof}
 For $N \geq 1$ let $\tau_N = \inf\{ t > 0 \ | \ \langle \mathcal{V}^n(t) \rangle \geq N \text{ or } \langle \mathcal{V}^n(t-) \rangle \geq N \}$.
 \\ \textit{(a)} Consider first $(\mathcal{R}^n(t),\mathcal{V}^n(t))$ with deterministic initial condition $(r_0,v_0) \in \R^{2dn}$. 
 Applying the Lemma \ref{EST:03}.(a) and then using \eqref{IPS:EQ02} together with conservation of kinetic energy gives
 \begin{align*}
   \E_{(r_0,v_0)}^n\left( \frac{\langle \mathcal{V}^n(t\wedge \tau_N)\rangle_{2p}}{n}\right)
   &\leq  \frac{\langle v_0\rangle_{2p}}{n} + C_p\frac{\langle v_0 \rangle_{2}}{n}\E_{(r_0,v_0)}^n\left(\int \limits_{0}^{t\wedge \tau_N} \frac{\langle \mathcal{V}^n(s)\rangle_{2p-2+\gamma}}{n}ds\right)
   \\ &\leq \frac{\langle v_0\rangle_{2p}}{n} + C_p\frac{\langle v_0 \rangle_{2}}{n} \int \limits_{0}^{t} \E_{(r_0,v_0)}^n\left( \frac{\langle \mathcal{V}^n(s \wedge \tau_N)\rangle_{2p}}{n}ds\right)^{1 - \frac{2-\gamma}{2p}}ds.
 \end{align*}
 Applying the Bihari-LaSalle inequality we obtain for some $C_p > 0$
 \[
   \E_{(r_0,v_0)}^n\left( \frac{\langle \mathcal{V}^n(t\wedge \tau_N)\rangle_{2p}}{n}\right)
   \leq C_p \frac{\langle v_0\rangle_{2p}}{n} + C_p \left( \frac{\langle v_0 \rangle_{2}}{n} \right)^{\frac{2p}{2-\gamma}} t^{\frac{2p}{2-\gamma}}.
 \]
 Taking $N \to \infty$ proves the assertion in the case of deterministic initial condition. 
 For the general case we use $\Pr_{\rho^{(n)}} = \int_{\R^{2dn}}\Pr_{(r_0,v_0)} d\rho^{(n)}(r_0,v_0)$ to obtain
 \begin{align*}
  \E^n\left( \frac{\langle \mathcal{V}^n(t)\rangle_{2p}}{n}\right)
   &\leq C_p \E^n \left( \frac{\langle \mathcal{V}^n(0)\rangle_{2p}}{n}\right) + C_p \E^n\left(\left( \frac{\langle \mathcal{V}^n(0) \rangle_{2}}{n} \right)^{\frac{2p}{2-\gamma}} \right) t^{\frac{2p}{2-\gamma}}
   \\ &\leq C_p \E^n \left( \frac{\langle \mathcal{V}^n(0)\rangle_{2p}}{n}\right) + C_p \E^n\left(\frac{\langle \mathcal{V}^n(0) \rangle_{\frac{4p}{2-\gamma}}}{n} \right) t^{\frac{2p}{2-\gamma}}.
 \end{align*}
 \textit{(b)} Proceed in the same way as in part (a) with the Gronwall lemma instead.
\end{proof}

\begin{Proposition}\label{IPS:LEMMA00HP}
 Suppose that the initial distribution $\rho^{(n)} \in \mathcal{P}(\R^{2dn})$ is symmetric. Then the following assertions hold:
 \begin{enumerate}
  \item[(a)] Suppose that $\gamma \in [0,2)$. Then for each $2p$ with $2p + \gamma \geq 4$ there exists a constant $C_p > 0$ such that 
  \[
   \E^n\left( \sup \limits_{s \in [0,t]} \langle \mathcal{V}_1^n(s) \rangle^{2p}\right) \leq C_p \E^n\left( \frac{\langle \mathcal{V}^n(0)\rangle_{\frac{2(2p+\gamma)}{2-\gamma}}}{n}\right)(1 + t^{\frac{2p+2}{2-\gamma}}),
  \]
  provided the right-hand side is finite.
  \item[(b)] Suppose that $\gamma = 2$. Then for each $p \geq 1$ there exists a constant $C_p > 0$ such that
  \[
   \E^n\left( \sup \limits_{s \in [0,t]}\langle \mathcal{V}_1^n(s)\rangle^{2p}\right) \leq C_p \E^n\left(\frac{\langle \mathcal{V}^n(0)\rangle_{2p + \gamma}}{n} e^{C_p t \frac{\langle \mathcal{V}^n(0)\rangle_2}{n} } \right)
  \]
 \end{enumerate}
\end{Proposition}
\begin{proof}
 Let us prove only part (a). By Lemma \ref{EST:03}.(b) and \eqref{IPS:EQ03} we obtain 
 \begin{align*}
  \E^n\left( \sup \limits_{s \in [0,t]} \langle \mathcal{V}_1^n(s) \rangle^{2p}\right)
  &\leq \E^n\left( \langle \mathcal{V}_1^n(0)\rangle^{2p}\right) + C \int \limits_{0}^{t}\E^n\left( \frac{\langle \mathcal{V}^n(s)\rangle_{2p+\gamma}}{n} \right)ds
  \\ &\leq  \E^n\left( \langle \mathcal{V}_1^n(0)\rangle^{2p}\right) + C t \E^n\left( \frac{\langle \mathcal{V}^n(0)\rangle_{2p+\gamma}}{n} \right) + C \E^n\left( \frac{\langle \mathcal{V}^n(0)\rangle_{\frac{2(2p+\gamma)}{2-\gamma}}}{n} \right)\int \limits_{0}^{t}s^{\frac{2p+\gamma}{2-\gamma}}ds
  \\ &\leq C \E^n\left( \frac{\langle \mathcal{V}^n(0)\rangle_{\frac{2(2p+\gamma)}{2-\gamma}}}{n} \right) (1 + t^{\frac{2p + 2}{2-\gamma}})
 \end{align*}
 where we have used $t^a \leq 1 + t^b$ for $t \geq 0$ and $0 < a < b$.
\end{proof}

\section{Particle approximation}
In this section we study the infinite particle limit $n \to \infty$ for the interacting particle system with generator \eqref{EQ:00}.
As a consequence we complete the proofs of Theorem \ref{TH:04}, Theorem \ref{TH:05} and Theorem \ref{TH:06}.

\subsection{Tightness and moment estimates}
Let $\mu_0$ be the initial condition as prescribed in Theorem \ref{TH:04}, Theorem \ref{TH:05} or Theorem \ref{TH:06}.
Define for each $n \geq 2$ by $\rho^{(n)} = \mu_0^{\otimes n}$ a family of probability measures on $\R^{2dn}$.
Let $(\mathcal{R}^n(t),\mathcal{V}^n(t))$ be the weak solutions to \eqref{IPS:EQ01} on $(\Omega_n, \F^n, \F_t^n, \Pr^n)$ with law $\Pr_{\rho^{(n)}}$, 
i.e. the unique solution to the martingale problem $(L, C_c^1(\R^{2dn}), \rho^{(n)})$. 
Write $\mathcal{X}^n = (\mathcal{X}_1^n, \dots, \mathcal{X}_n^n)$ with $\mathcal{X}_k^n = (\mathcal{R}_k^n, \mathcal{V}_k^n)$.
\begin{Lemma}\label{MOMENT:EST}
 The following moment estimates hold.
 \begin{enumerate}
  \item[(a)] If $\gamma \in (-1,0]$, then for each $T > 0$
  \[
   \E^n\left( \sup \limits_{t \in [0,T]} \langle \mathcal{V}_1^n(t) \rangle^{2+\gamma} \right) \leq \int \limits_{\R^{2d}}\langle v \rangle^{2+\gamma}\mu_0(r,v) e^{C_p T}.
  \]
  \item[(b)] If $\gamma \in (0,2)$, then for each $T > 0$
  \[
   \sup \limits_{t \in [0,T]}\E^n\left( \langle \mathcal{V}_1^n(t)\rangle^{4} \right) + \E^n\left(\sup \limits_{t \in [0,T]} \langle \mathcal{V}_1^n(t)\rangle^{1+\gamma}\right)
   \leq e^{C_{p,\gamma}T} \int \limits_{\R^{2d}}\langle v \rangle^{\frac{2}{2-\gamma}\max\{4, 1+2\gamma \}}d\mu_0(r,v).
  \]
  \item[(c)] If $\gamma = 2$, then for each $p \geq 1$ and $T > 0$ there exists a constant $n_0(T,p) \geq 2$ such that
  \[
   \sup \limits_{n \geq n_0(T,p)}\ \E^n\left( \sup \limits_{t \in [0,T]}\langle \mathcal{V}_1^n(t)\rangle^{2p}\right) < \infty.
  \]
 \end{enumerate}
\end{Lemma}
\begin{proof}
 Assertion (a) is a particular case of Proposition \ref{IPS:TH01SP} (see Remark \ref{REMARK}).
 Concerning assertion (b) we obtain from Theorem \ref{IPS:TH01HP}.(a)
 \begin{align*}
   \sup \limits_{t \in [0,T]}\E^n\left( \langle \mathcal{V}_1^n(t)\rangle^{4} \right)
  = \sup \limits_{t \in [0,T]} \E^n \left( \frac{\langle \mathcal{V}^n(t) \rangle_4}{n} \right)
  \leq C_p \int \limits_{\R^{2d}} \langle v \rangle^{\frac{8}{2-\gamma}} \mu_0(dr,dv) \left( 1 + t^{\frac{4}{2-\gamma}} \right).
 \end{align*}
 For the second term we apply Proposition \ref{IPS:LEMMA00HP}.(a), for $2p^* = \max\{ 1 + 2\gamma, 4-\gamma\}$, which yields
 \begin{align*}
  \E^n\left(\sup \limits_{t \in [0,T]} \langle \mathcal{V}_1^n(t)\rangle^{1+\gamma}\right)
 &\leq \E^n\left(\sup \limits_{t \in [0,T]} \langle \mathcal{V}_1^n(t)\rangle^{2p^*}\right)
 \\ &\leq C_p\int \limits_{\R^{2d}}\langle v \rangle^{\frac{2}{2-\gamma}\max\{4, 1+2\gamma \}}\mu_0(dr,dv) \left( 1 + t^{\frac{2p^* + 2}{2-\gamma}}\right).
 \end{align*}
 It remains to prove assertion (c). By Lemma \ref{EST:03}.(b) together with \eqref{IPS:EQ03} we obtain
 \begin{align*}
  &\ \E^n\left( \sup \limits_{s \in [0,t]}\langle\mathcal{V}_1^n(s)\rangle^{2p} \right) \leq \E^n\left( \frac{\langle \mathcal{V}^n(0)\rangle_{2p}}{n}  \right) 
  + C \int \limits_{0}^{t}\E^n\left( \frac{\langle \mathcal{V}^n(s) \rangle_{2p + \gamma}}{n}\right)ds.
 \end{align*}
 Hence it remains to show that for each $q \geq 2$ we find $n_0 = n_0(t,q) \geq 2$ such that
 \begin{align}\label{EQ:22}
  \sup \limits_{n \geq n_0}\sup \limits_{s \in [0,t]} \E^n\left( \frac{\langle \mathcal{V}^n(s) \rangle_{2q}}{n} \right) < \infty.
 \end{align}
 By Theorem \ref{IPS:TH01HP} we obtain for some constant $C_q > 0$ and $q \geq 2$
 \begin{align*}
  \sup \limits_{t \in [0,T]}\E^n\left( \frac{\langle \mathcal{V}^n(t) \rangle_{2q}}{n} \right) 
  &\leq \E^n\left( \frac{\langle \mathcal{V}(0) \rangle_{2q}}{n} e^{C_q \frac{\langle \mathcal{V}(0) \rangle_2}{n} T} \right)
  \\ &= \frac{1}{n} \sum \limits_{k=1}^{n}\int \limits_{\R^{2dn}}\langle v_k \rangle^{2q}e^{C_q \frac{\langle v_k \rangle^2}{n}T} \prod \limits_{j \neq k} e^{C_q \frac{\langle v_j \rangle^2}{n}T} d\mu_0^{\otimes n}(r,v)
  \\ &= \int \limits_{\R^{2d}} \langle v \rangle^{2q} e^{C_q \frac{\langle v \rangle^2}{n}T}d\mu_0(r,v) \left( \int \limits_{\R^{2d}}e^{C_q \frac{\langle v \rangle^2}{n}T} d\mu_0(r,v)\right)^{n-1}
  \\ &\leq A_q \left( \int \limits_{\R^{2d}}e^{C_q' \frac{\langle v \rangle^2}{n}T} d\mu_0(r,v)\right)^{n}
 \end{align*}
 for some constants $A_q, C_q' > 0$.
 Let $n_0 = n_0(T,q) \geq 2$ be such that such that $C_q' T \leq a n_0$, then 
\[
 n \left| e^{\frac{C_q' T}{n}\langle v \rangle^2} - 1 \right| 
 \leq C_q' T e^{\frac{C_q' T}{n} \langle v \rangle^2} \langle v \rangle^2 \leq C_q' T e^{a \langle v \rangle^2} \langle v \rangle^2, \ \ n \geq n_0,
\]
and we may assume without loss of generality that the right-hand side is integrable w.r.t. $\mu_0$.
Hence we obtain
\begin{align*}
 \left( \int \limits_{\R^{2d}}e^{C_q' \frac{\langle v \rangle^2}{n}T} d\mu_0(r,v)\right)^{n}
 &= \left( 1 + \frac{1}{n}\int \limits_{\R^{2d}}n \left(e^{\frac{C_q' T}{n}\langle v \rangle^2} - 1\right)d\mu_0(r,v) \right)^n
 \\ &\leq \exp\left( C_q' T \int \limits_{\R^{2d}}\langle v \rangle^2 e^{a \langle v \rangle^2}d\mu_0(r,v) \right) < \infty.
\end{align*}
 This proves \eqref{EQ:22}.
\end{proof}
Recall that $\mu^{(n)} = \frac{1}{n}\sum_{k=1}^{n}\delta_{\mathcal{X}_k^n}$ is the empirical distribution and let $\pi^{(n)}$ be its law.
\begin{Proposition}\label{RELCOMP}
 The sequence $(\pi^{(n)})_{n \geq 2}$ is relatively compact.
\end{Proposition}
\begin{proof}
 It follows by standard theory that $\pi^{(n)}$ is relatively compact iff the first coordinate $\mathcal{X}_1$ is tight (see \cite[Proposition 2.2.(ii)]{S91}),
 see also Lemma \ref{EXCH}. We seek to apply the Aldous criterion. It is not difficult to see that
 \[
  \sup \limits_{t \in [0,T]} \sup \limits_{n \in \N} \E^n( \langle \mathcal{R}_1^n(t) \rangle^{\e} + \langle \mathcal{V}_1^n(t) \rangle ) < \infty, \ \ \forall T > 0,
 \]
 where $\e$ is such that $\int_{\R^{2d}}|r|^{\e} \mu_0(dr,dv) < \infty$.
 Next, let $S_n, T_n$ be $\F_t^n$ stopping times with $S_n \leq T_n \leq \delta + S_n$ and $S_n,T_n \leq M \in \N$. Take any $\delta \in (0,1)$. Then 
 \[
 \E^{n}\left( |\mathcal{R}_1^{n}(S_n) - \mathcal{R}_1^{n}(T_n)| \right) 
 \leq \E^n \left(\int \limits_{S_n}^{T_n} |\mathcal{V}_1^n(s)| ds \right)
 \leq \delta \E^n\left( \sup_{s \in [0,M]}|\mathcal{V}_1^n(s)| \right)
 \]
 and, by exchangeability, we obtain
 \begin{align*}
  \E^{n}\left( |\mathcal{V}_1^{n}(S_n) - \mathcal{V}_1^{n}(T_n)|\  \right)
   &\leq \frac{C}{n}\sum \limits_{k=1}^{n} \E^n\left( \int \limits_{S_n}^{T_n}\int \limits_{E} |\alpha(\mathcal{V}_1^{n}(s), \mathcal{V}_k^{n}(s),\theta,\xi)|\sigma(|\mathcal{V}_1^{n}(s)- \mathcal{V}_k^{n}(s)|)| dQ(\theta)d\xi ds \right)
  \\ &\leq \frac{C}{n} \sum \limits_{k=1}^{n}\E^n\left( \int \limits_{S_n}^{T_n} \left( \langle \mathcal{V}_1^{(n)}(s)\rangle^{1+\gamma} + \langle \mathcal{V}_k^{(n)}(s)\rangle^{1+\gamma} \right) ds \right)
  \\ &\leq C \delta \E^{n}\left( \sup \limits_{s \in [0,M]}\langle \mathcal{V}_1^n(s)\rangle^{1+\gamma} \right).
 \end{align*}
 The assertion follows now from the moment estimates given in Lemma \ref{MOMENT:EST}.
\end{proof}
For $\nu \in \mathcal{P}(D(\R_+;\R^{2d}))$ let $\nu_t \in \mathcal{P}(\R^{2d})$ be the time-marginal at time $t \geq 0$ and for $q \geq 0$ set
$\| \nu_t \|_q := \int_{\R^{2d}} \langle v \rangle^q d\nu_t(r,v)$.
The desired moment estimates are deduced from the following Lemma of Fatou together with the estimates from previous section.
\begin{Lemma}\label{IPS:LEMMA05} 
 Let $\pi^{(\infty)}$ be any accumulation point of $\pi^{(n)}$. Then for any $2p \geq 1$
 \[
  \int \limits_{\mathcal{P}(D(\R_+;\R^{2d}))} \| \nu_u \|_{2p} \ d\pi^{(\infty)}(\nu) 
  \leq \liminf \limits_{n \to \infty} \E^n\left( \langle V_1^n(u) \rangle^{2p}\right), \ \ u \geq 0
 \]
 and moreover
 \[
   \int \limits_{\mathcal{P}(D(\R_+;\R^{2d}))} \sup \limits_{u \in [0,t]} \| \nu_u \|_{2p} \ d\pi^{(\infty)}(\nu) 
   \leq \liminf \limits_{n \to \infty} \sup \limits_{u \in [0,t]}  \E^n\left( \langle V_1^n(u) \rangle^{2p}\right), \ \ t \geq 0,
 \]
 provided the right-hand side is finite.
\end{Lemma}

\subsection{Convergence of martingale problems}
Below we prove that any limit point solves the martingale problem for the Enskog process.
\begin{Theorem}\label{TH:10}
 Let $\pi^{(\infty)}$ be any accumulation point of $\pi^{(n)}$.
 Then $\pi^{(\infty)}$-a.a. $\nu \in \mathcal{P}(D(\R_+;\R^{2d}))$ solve the martingale problem $(A(\nu_s), C_c^1(\R^{2d}), \mu_0)$.
\end{Theorem}
The rest of this section is devoted to the proof of this Theorem.
 Observe that for given $\nu \in \mathcal{P}(D(\R_+;\R^{2d}))$ the complement of $D_{\nu} = \{ t > 0 \ | \ \nu( x(t) = x(t-) ) = 1 \}$ is at most countable
 and the coordinate function $x \longmapsto x(t)$ is $\nu$-a.s. continuous for any $t \in D_{\nu}$.
 Moreover also the complement of $D(\pi^{(\infty)}) = \{ t > 0 \ | \ \pi^{(\infty)}( \nu \ | \ t \in D_{\nu} ) = 1 \}$ is at most countable.
 For simplicity of notation we denote the subsequence of $\pi^{(n)}$ converging to $\pi^{(\infty)}$ again by $\pi^{(n)}$.
 Fix any $0 \leq s_1,\dots, s_m \leq s \leq t \in D(\pi^{(\infty)})$, $g_1,\dots, g_m \in C_b(\R^{2d})$, $m \in \N$ and $\psi \in C_c^1(\R^{2d})$. 
 Define 
\begin{align*}
 F(\nu) = \int \limits_{D(\R_+; \R^{2d})} \left( \psi(x(t)) - \psi(x(s)) - \int \limits_{s}^{t}(A(\nu_u)\psi)(x(u))du \right) \prod \limits_{k=1}^{m}g_k(x(s_k)) d\nu(x).
\end{align*}
It follows from \eqref{EXISTENCE:01} that we can find a constant $C > 0$ (independent of $\nu$) such that 
\begin{align*} 
 |F(\nu)| \leq C\sup \limits_{u \in [s,t]} \| \nu_u \|_{1+\gamma^+}, \ \ \ \forall \nu \in \mathcal{P}(D(\R_+;\R^{2d}))  
\end{align*}
and hence Lemma \ref{IPS:LEMMA05} applied for 
$2p = 1+\gamma^+$ if $\gamma \in [0,2]$ and to $2p = 2+\gamma$ if $\gamma \in (-1,0)$ gives
\[
 \int \limits_{\mathcal{P}(D(\R_+;\R^{2d}))} |F(\nu)| \ d\pi^{(\infty)}(\nu) < \infty, \ \ \forall t > 0.
\]
It is clear that $\nu$ is a solution to the martingale problem posed by \eqref{SDE:ENSKOG} if $F(\nu) = 0$.
Hence it suffices to prove the following assertions:
\begin{enumerate}
 \item[(a)] $\lim_{n \to \infty} \int_{\mathcal{P}(D(\R_+;\R^{2d}))} |F(\nu)|^2 d\pi^{(n)}(\nu) = 0$,
 \item[(b)] $\lim_{n \to \infty} \int_{\mathcal{P}(D(\R_+;\R^{2d}))}|F(\nu)| d\pi^{(n)}(\nu) = \int_{\mathcal{P}(D(\R_+;\R^{2d}))}|F(\nu)| d\pi^{(\infty)}(\nu)$.
\end{enumerate}
Let us first prove assertion (a).
\begin{Proposition}
 Assertion (a) is satisfied.
\end{Proposition}
\begin{proof}
 Let $\widetilde{N}$ be the compensated Poisson random measure with compensator \eqref{EQ:06},
 \[
  G(r,v,z,\theta,\xi,l,l') = (e_l - e_{l'})\alpha(v_l,v_{l'}, \theta,\xi) \1_{[0, \sigma(|v_l - v_{l'}|)\beta(r_l - r_{l'})]}(z),
 \]
 where $(r,v) \in \R^{2dn}, \ z \in \R_+, \ (\theta,\xi,l,l') \in E := \Xi \times \{1,\dots,n\}^2$ is defined as in Theorem \ref{IPS:TH00} and,
 \begin{align*}
  M_{s,t}^{n,k} = \int \limits_{s}^{t}\int \limits_{E}\left( \psi(\mathcal{R}_k^{n}(u), \mathcal{V}_k^{n}(u-) + G_k) - \psi(\mathcal{R}_k^{n}(u),\mathcal{V}_k^{n}(u-))\right)d\widetilde{N}(u,l,l',\theta,\xi,z),
 \end{align*}
 where $G_k = G_k( \mathcal{R}^n(u), \mathcal{V}^n(u-),z,\theta,\xi,l,l')$.
 Then a short computation shows that
 \begin{align*}
  F(\mu^{(n)}) = \frac{1}{n}\sum_{k=1}^{n}M_{s,t}^{n,k}\prod_{j=1}^{m}g_j(\mathcal{X}_k^{n}(s_j)).
 \end{align*}
 Indeed, it holds that
 \begin{align*}
  &\ (A(\mu_u^{(n)})\psi)(\mathcal{X}_k^{n}) = \mathcal{V}_k^{n} \cdot (\nabla_r \psi)(\mathcal{X}_k^{n})
  \\ &+ \frac{1}{n}\sum \limits_{j=1}^{n}\sigma(|\mathcal{V}_k^{n}-\mathcal{V}_j^{n}|)\beta(\mathcal{R}_k^{n} - \mathcal{R}_j^{n})\int \limits_{\Xi}\left( \psi(\mathcal{R}_k^{n},  \mathcal{V}_k^{n} + \alpha(\mathcal{V}_k^{n}, \mathcal{V}_j^{n},\theta,\xi) ) - \psi(\mathcal{R}_k^{n},\mathcal{V}_k^{n})\right)Q(d\theta)d\xi,
 \end{align*}
 and from the It\^{o} formula one immediately obtains
 \begin{align*}
  \psi(\mathcal{X}_k^{n}(t)) - \psi(\mathcal{X}_k^{n}(s)) &= \int \limits_{s}^{t}(A(\mu_u^{(n)})\psi)(\mathcal{X}_k^{n}(u))du + M_{s,t}^{n,k}.
 \end{align*}
 The Doob-Meyer process of $M_{s,t}^{n,k}$ satisfies
 \begin{align*}
  \langle M_{s,t}^{n,k} \rangle &= \frac{1}{2n} \sum \limits_{j=1}^{n} \int \limits_{s}^{t} \int \limits_{\Xi} \left[ \psi( \mathcal{R}_k^n, \mathcal{V}_k^n + \alpha(\mathcal{V}_k^n, \mathcal{V}_{j}^n,\theta,\xi)) - \psi(\mathcal{R}_k^n, \mathcal{V}_k^n) \right]^2 \sigma(|\mathcal{V}_k^n - \mathcal{V}_j^n|) \beta(\mathcal{R}_k^n - \mathcal{R}_j^n) Q(d\theta)d\xi du
 \\ &\ + \frac{1}{2n} \sum \limits_{j=1}^{n} \int \limits_{s}^{t} \int \limits_{\Xi} \left[ \psi( \mathcal{R}_k^n, \mathcal{V}_k^n - \alpha(\mathcal{V}_k^n, \mathcal{V}_{j}^n,\theta,\xi)) - \psi(\mathcal{R}_k^n, \mathcal{V}_k^n) \right]^2 \sigma(|\mathcal{V}_k^n - \mathcal{V}_j^n|) \beta(\mathcal{R}_k^n - \mathcal{R}_j^n) Q(d\theta)d\xi du
 \\ &\leq \frac{C}{n} \sum \limits_{j=1}^{n} \int \limits_{s}^{t} \int \limits_{\Xi} |\alpha(\mathcal{V}_k^n, \mathcal{V}_{j}^n,\theta,\xi))|^2\sigma(|\mathcal{V}_k^n - \mathcal{V}_j^n|)  Q(d\theta)d\xi du
 \\ &\leq \frac{C}{n} \sum \limits_{j=1}^{n} \int \limits_{s}^{t} \left( \langle \mathcal{V}_k^n(u) \rangle^{2+\gamma} + \langle \mathcal{V}_j^n(u) \rangle^{2+\gamma} \right)) du
 \end{align*}
 from which we obtain $\E^n( \langle M_{s,t}^{n,k} \rangle) \leq \frac{C}{n}$, for all $k = 1, \dots, n$.
 The covariance is, for $k \neq j$, given by
 \begin{align*}
  \langle M_{s,t}^{n,k}, M_{s,t}^{n,j} \rangle = \frac{1}{2n} \int \limits_{s}^{t} \int \limits_{\Xi} 
 &\left[ \psi( \mathcal{R}_k^n, \mathcal{V}_k^n + \alpha(\mathcal{V}_k^n, \mathcal{V}_j^n, \theta,\xi)) - \psi(\mathcal{R}_k^n, \mathcal{V}_k^n) \right] 
 \\ &\left[ \psi( \mathcal{R}_j^n, \mathcal{V}_j^n - \alpha(\mathcal{V}_k^n, \mathcal{V}_j^n, \theta,\xi)) - \psi(\mathcal{R}_j^n, \mathcal{V}_j^n) \right] 
 \sigma(|\mathcal{V}_k^n - \mathcal{V}_j^n|) \beta(\mathcal{R}_k^n - \mathcal{R}_j^n) Q(d\theta)d\xi du
 \end{align*}
 and hence by a similar computation
 \begin{align*}
  \E^n( | \langle M_{s,t}^{n,k}, M_{s,t}^{n,j} \rangle | )
 &\leq \frac{C}{n} \E^n\left( \int \limits_{s}^{t}\int \limits_{\Xi} |\alpha(\mathcal{V}_k^n(u), \mathcal{V}_j^n(u),\theta,\Xi)|^2 \sigma(|\mathcal{V}_k^n(u) - \mathcal{V}_j^n(u)|) Q(d\theta)d\xi du \right)
 \leq \frac{C}{n}.
 \end{align*}
 This implies that
 \begin{align*}
  &\ \int \limits_{\mathcal{P}(D(\R_+;\R^{2d}))}|F(\nu)|^2 d\pi^{(n)}(\nu) = \E^{n}\left( |F(\mu^{(n)})|^2 \right)
  \\ &= \frac{1}{n^2}\sum \limits_{k=1}^{n}\E^{n}\left( (M_{s,t}^{n,k})^2 \prod_{j=1}^{m}g_j(\mathcal{X}_k^{n}(s_j))^2 \right) 
      + \frac{1}{n^2} \sum \limits_{k \neq j}\E^{n}\left( M_{s,t}^{n,k}M_{s,t}^{n,j}\prod_{l_1,l_2=1}^{m}g_{l_1}(\mathcal{X}_k^{n}(s_{l_1}))g_{l_2}(\mathcal{X}_j^{n}(s_{l_2}))\right)
  \\ &\leq \frac{C}{n^2}\sum \limits_{k=1}^{n}\E^{n}\left( \langle M_{s,t}^{n,k} \rangle \right) + 
  \frac{C}{n^2}\sum \limits_{k \neq j}\E^n( | \langle M_{s,t}^{n,k}, M_{s,t}^{n,j} \rangle| )
   \leq \frac{C}{n},
 \end{align*}
 i.e. assertion (a) is proved.
\end{proof}
The proof of Theorem \ref{TH:10} is completed once we have shown the next proposition.
\begin{Proposition}
 Assertion (b) is satisfied.
\end{Proposition}
\begin{proof}
 We proceed in 3 steps.

 \textit{Step 1.} Let $g$ be a smooth function on $\R_+$ such that $\1_{[0,1]} \leq g \leq \1_{[0,2]}$.  
 For $R > 0$ and $\nu \in \mathcal{P}(\R^{2d})$ let (see \eqref{ENSKOG:OP})
 \[
  (A_R(\nu)\psi)(r,v) = \int \limits_{\R^{2d}}g\left( \frac{|u|^2}{R^2}\right)(\mathcal{A}\psi)(r,v;q,u)d\nu(q,u).
 \]
 Similarly to Proposition \ref{FPE:LEMMA01} and Proposition \ref{FPE:PROP00} we can show the following 
 \begin{itemize}
  \item $A_R(\nu)\psi$ is jointly continuous in $(\nu,r,v) \in \mathcal{P}(\R^{2d}) \times \R^{2d}$, where $\mathcal{P}(\R^{2d})$ 
   is endowed the topology of weak convergence.
  \item There exists a constant $C > 0$ such that 
  \begin{align}\label{EQ:23}
   |A_R(\nu)\psi(r,v)| \leq \| v \cdot \nabla_r \psi\|_{\infty} + C( 1+ R^2)^{\frac{1+\gamma}{2}}.
  \end{align}
 \end{itemize}
 Moreover we have for some generic constant $C > 0$
 \begin{align*}
   |A_R(\nu)\psi(r,v) - A(\nu)\psi(r,v)| 
  &\leq \int \limits_{\R^{2d}}\left(1 - g\left( \frac{|u|^2}{R^2}\right)\right) |\mathcal{A}\psi(r,v;q,u)| d\nu(q,u)
   \\ &\leq \int \limits_{\R^{2d}}\1_{\{ |u| > R \}} |\mathcal{A}\psi(r,v;q,u)| d\nu(q,u)
   \\ &\leq C \int \limits_{\R^{2d}}\langle u \rangle^{1+\gamma}\1_{\{ |u| > R \}} d\nu(q,u)
   \leq \frac{C}{R} \| \nu \|_{2+\gamma}
 \end{align*}
 where we have used $|\mathcal{A}\psi(r,v;q,u)| \leq C \langle u \rangle^{1+ \gamma}$ (see \eqref{EQ:14}).
 \\ \textit{Step 2.} Define for $\nu \in \mathcal{P}(D(\R_+;\R^{2d}))$ and $x \in D(\R_+;\R^{2d})$
 \begin{align*}
  H_R(\nu; x) := \left( \psi(x(t)) - \psi(x(s)) - \int \limits_{s}^{t}(A_R(\nu_u)\psi)(x(u))du \right) \prod \limits_{k=1}^{m}g_k(x(s_k)).
 \end{align*}
 Then by \eqref{EQ:23} $H_R$ is bounded. Let $F_R(\nu) = \int_{D(\R_+;\R^{2d})}H_R(\nu;x) d\nu(x)$.
 Using $y = (y_1,y_2) \in D(\R_+;\R^{2d})$ we get
 \begin{align*}
  \int \limits_{s}^{t}(A_R(\nu_u)\psi)(x(u))du &= \int \limits_{s}^{t}\int \limits_{\R^{2d}}g\left( \frac{|u'|^2}{R^2}\right) (\mathcal{A}\psi)(x(u);q,u')d\nu_u(q,u')du
  \\ &= \int \limits_{D(\R_+;\R^{2d})} \int \limits_{s}^{t} g\left(\frac{|y_2(u)|}{R^2}\right) (\mathcal{A}\psi)(x(u);y(u)) du d\nu(y)
 \end{align*}
 and hence
 \begin{align*}
  F_R(\nu) &= \int \limits_{D(\R_+;\R^{2d})}(\psi(x(t)) - \psi(x(s))) d\nu(x) 
  \\ &- \int \limits_{D(\R_+;\R^{2d})^2} \int \limits_{s}^{t} g\left(\frac{|y_2(u)|}{R^2}\right) (\mathcal{A}\psi)(x(u);y(u)) \prod_{k=1}^{m}g_k(x(s_k))du d\nu(y)d\nu(x).
 \end{align*}
 This shows that $F_R(\nu)$ is bounded and continuous for $\pi^{(\infty)}$-a.a. $\nu$ w.r.t. weak convergence.
 \\ \textit{Step 3.} Write
 \begin{align*}
  &\ \left| \int \limits_{\mathcal{P}(D(\R_+;\R^{2d}))}|F(\nu)| d\pi^{(n)}(\nu) - \int \limits_{\mathcal{P}(D(\R_+;\R^{2d}))}|F(\nu)| d\pi^{(\infty)}(\nu) \right|
  \\ &\leq \left| \int \limits_{\mathcal{P}(D(\R_+;\R^{2d}))}|F(\nu)| d\pi^{(n)}(\nu) - \int \limits_{\mathcal{P}(D(\R_+;\R^{2d}))}|F_R(\nu)| d\pi^{(n)}(\nu) \right|
  \\ &\ \ \ + \left| \int \limits_{\mathcal{P}(D(\R_+;\R^{2d}))}|F_R(\nu)| d\pi^{(n)}(\nu) - \int \limits_{\mathcal{P}(D(\R_+;\R^{2d}))}|F_R(\nu)| d\pi^{(\infty)}(\nu) \right|
  \\ &\ \ \ + \left| \int \limits_{\mathcal{P}(D(\R_+;\R^{2d}))}|F_R(\nu)| d\pi^{(\infty)}(\nu) - \int \limits_{\mathcal{P}(D(\R_+;\R^{2d}))}|F(\nu)| d\pi^{(\infty)}(\nu) \right|
  \\ &= I_1 + I_2 + I_3
 \end{align*}
 Step 2 implies that the second term tends for each fixed $R$ to zero as $n \to \infty$.
 Using Step 1 we obtain for $T > t$ and some constant $C > 0$
 \begin{align*}
  |F(\nu) - F_R(\nu)| \leq C \int \limits_{s}^{t}\int \limits_{D(\R_+;\R^{2d})} | (A(\nu_u)\psi)(x(u)) - (A_R(\nu_u)\psi)(x(u))| d\nu(x)du
  \leq \frac{C}{R} \int \limits_{s}^{t} \| \nu_u\|_{2+\gamma} du.
 \end{align*}
 In view of Lemma \ref{MOMENT:EST} we find a constant $C' > 0$ such that
 $\limsup_{n \to \infty}(I_1 + I_3) \leq \frac{C'}{R}$, which proves the assertion.
\end{proof}

\subsection{From nonlinear martingale problem to SDE}
It remains to show that any solution to the martingale problem obtained from the particle approximation can be represented 
by a weak solution to \eqref{SDE:ENSKOG}.
\begin{Lemma}\label{REPRESENTATION}
 Let $\nu \in \mathcal{P}(D(\R_+; \R^{2d}))$ be a solution to the (nonlinear) martingale problem $(A(\nu_s), C_b^1(\R^{2d}), \mu_0)$
 satisfying
 \begin{align}\label{EQ:107}
  \sup \limits_{t \in [0,T]} \| \nu_t \|_{1+\gamma} < \infty, \qquad T > 0.
 \end{align}
 Then there exists an Enskog process $(R,V)$ with law given by $\nu$.
\end{Lemma}
\begin{proof}
 First observe that we can find a measurable process $(q_s(\eta),u_s(\eta))$ on some probability space $(\mathcal{X},d\eta)$ such that
 $\mathcal{L}(q_s,u_s) = \nu_s$ for $s \geq 0$.
Let $\widehat{\alpha}(v,r,u,q,\theta,\xi,z)$ be given as in \eqref{SDE:ENSKOG} and
\[
 b(s,v,r) = \begin{pmatrix} v \\ \int \limits_{\Xi \times \mathcal{X} \times \R_+} \widehat{\alpha}(v,r,u_s(\eta),q_s(\eta), \theta,\xi,z) Q(d\theta)d\xi d\eta dz \end{pmatrix}.
\]
Then it is not difficult to see that
$|b(s,v,r)|$, $\int_{\Xi \times \mathcal{X} \times \R_+} |\widehat{\alpha}(v,r,u_s(\eta),q_s(\eta), \theta,\xi,z)|^2 Q(d\theta)d\xi d\eta dz$
are locally bounded on $\R_+ \times \R^{2d}$. Moreover for each $\psi \in C_c^1(\R^{2d})$ and $T > 0$ we find a constant $C = C(\psi, T) > 0$ such that 
$\| A(\nu_t)\psi \|_{\infty} \leq C \sup_{t \in [0,T]} \| \nu_t \|_{1+\gamma}$. 
Hence we may apply \cite[Theorem A.1]{HK90} to find a stochastic basis $(\Omega,\F, \F_t, \Pr)$, a poisson random measure $dN$ 
with compensator $d\widehat{N} =  Q(d\theta)d\xi d\eta dz$ and an $\F_t$-adapted  c\`{a}dl\`{a}g process $(X_t)_{t \geq 0} = (R_t,V_t)_{t \geq 0}$ on $(\Omega,\F, \F_t, \Pr)$ such that $\mathcal{L}(R,V) = \nu$ and 
   \begin{align}\label{INTCOMP}
    X_t &= X_0 + \int \limits_{0}^{t}b(s,X_s)ds
  + \int \limits_{0}^{t}\int \limits_{\Xi \times \mathcal{X} \times \R_+} \widehat{\alpha}(V_{s-},R_s, u_{s}(\eta),q_s(\eta),\theta,\xi,z)d\widetilde{N}(s,\theta,\xi,\eta,z).
   \end{align}
 Using \eqref{EQ:107} we get 
\[
\int_0^t \int \limits_{\Xi \times \mathcal{X} \times \R_+} \E\left(  |\widehat{\alpha}(V_{s-},R_s, u_{s-}(\eta),q_s(\eta),\theta,\xi,z)| \right ) d\widehat{N}(s,\theta,\xi,\eta,z) < \infty
\]
 and hence omitting the compensation in \eqref{INTCOMP} and modifying the drift $b$ appropriately shows that $X_t$ satisfies \eqref{SDE:ENSKOG}.
\end{proof}

\section*{Acknowledgements}

We would like thank Errico Presutti and Mario Pulvirenti for usefull discussions related to this work.

\appendix
\section*{Appendix}
\renewcommand{\thesection}{A} 
\setcounter{Theorem}{0}

The following is a nonlinear generalization of the Gronwall lemma, also known as the Bihari-LaSalle inequality.
\begin{Lemma}\label{LASALLE}
  Let $f: \R_+ \longrightarrow \R_+$ be measurable and suppose that 
 \[
  f(t) \leq f(0) + K \int \limits_{0}^{t}f(s)^{1-\alpha} ds, \ \ t \geq 0
 \]
 for some $K \geq 0$ and $\alpha \in (0,1)$. Then for any $t \geq 0$
 \[
  f(t) \leq \left( f(0)^{\alpha} + \alpha K t \right)^{1 / \alpha} \leq 2^{1/ \alpha-1}f(0) + \frac{\left( 2 \alpha K  \right)^{1 / \alpha} }{2} t^{1 / \alpha}.
 \]
\end{Lemma}

\begin{footnotesize}

\bibliographystyle{alpha}
\bibliography{Bibliography}

\end{footnotesize}

\end{document}